\documentclass{CSML}

\usepackage{lastpage}
\newcommand{\dOi}{13(3:19)2017}
\lmcsheading{}{1--\pageref{LastPage}}{}{}%
{Sep.~23,~2016}{Sep.~04, 2017}{}

\usepackage[T1]{fontenc}

\usepackage{enumerate}
\usepackage{hyperref}
\hypersetup{hidelinks}
\usepackage{colortbl}
\usepackage[usenames,dvipsnames,svgnames,table]{xcolor}
\usepackage{xspace}
\usepackage{gastex}
\usepackage{float}

\theoremstyle{plain}
\theoremstyle{plain}\newtheorem{theorem}[thm]{Theorem}
\theoremstyle{plain}
\theoremstyle{plain}\newtheorem{example}[thm]{Example}
\theoremstyle{plain}\newtheorem{lemma}[thm]{Lemma}
\theoremstyle{plain}
\theoremstyle{plain}\newtheorem{proposition}[thm]{Proposition}

\def\eg{{\em e.g.}}
\def\cf{{\em cf.}}

\begin{document}

\title{Path Checking for MTL and TPTL over Data Words}

\author[Shiguang Feng]{Shiguang Feng\rsuper{a,\clubsuit}}	
\address{\lsuper{a}Institut f\"ur Informatik, Universit\"at Leipzig, Germany}	
\thanks{\lsuper{\clubsuit}The author is supported by the German Research Foundation (DFG), GRK 1763.}	

\author[Markus Lohrey]{Markus Lohrey\rsuper{b}}	
\address{\lsuper{b}Department f\"ur Elektrotechnik und Informatik,
  Universit\"at Siegen, Germany}	

\author[Karin Quaas]{Karin Quaas\rsuper{a,\spadesuit}}	
\thanks{\lsuper{\spadesuit}The author is supported by DFG, projects QU~316/1-1 and QU~316/1-2.}	



\titlecomment{An extended abstract of this paper has been published at \emph{Developments in Language Theory} (DLT), 2015.}


\definecolor{Gray}{gray}{0.9}
\newcommand{\Op}{\mathsf{Op}}

\newcommand{\op}{\mathsf{op}}
\newcommand{\zero}{\mathsf{zero}}
\newcommand{\add}{\mathsf{add}}
\newcommand{\halt}{\mathsf{halt}}

\newcommand{\turing}{\mathbf{TM}}
\newcommand{\rmodels}{\models^{\textup{rel}}}
\newcommand{\fraisse}{Fra\"{i}ss\'{e}}
\newcommand{\game}{\mathbf{G}}
\newcommand{\tabincell}[2]{\begin{tabular}{@{}#1@{}}#2\end{tabular}}
\newcommand{\witness}{\mathsf{witness}}
\newcommand{\recurrent}{\mathsf{recurrent}}

\newcommand{\fltl}{\freeze}
\def\Karin#1{{\sf \fbox{$\spadesuit$ Karin: #1 $\spadesuit$}}}
\newcommand{\expspace}{\textup{EXPSPACE}}
\newcommand{\aptime}{\textup{APTIME}}
\newcommand{\pspace}{\textup{PSPACE}}
\newcommand{\ptime}{\textup{P}} 
\newcommand{\ls}{\textup{LOGSPACE}}
\newcommand{\alogspace}{\textup{ALOGSPACE}}
\newcommand{\np}{\textup{NP}}
\newcommand{\mync}{\textup{NC}}

\newcommand{\hyperack}{\textup{HyperAckermann}}
\newcommand{\ackermann}{\textup{Ackermann}}

\newcommand{\F}{\finally}
\newcommand{\G}{\glob}
\newcommand{\unary}{\textup{Unary}}
\newcommand{\pos}{\textup{Positive}}
\newcommand{\pun}{\textup{PosUnary}}

\newcommand{\puretptl}{\tptlpure}
\newcommand{\tptlpure}{\pure\tptl}
\newcommand{\tptlpun}{\pun\tptl}
\makeatletter
\newcommand*{\defeq}{\mathrel{\rlap{%
                     \raisebox{0.3ex}{$\m@th\cdot$}}%
                     \raisebox{-0.3ex}{$\m@th\cdot$}}%
                     =}
\makeatother

\makeatletter
\newcommand*{\ndefeq}{\mathrel{\rlap{%
                     \raisebox{0.3ex}{$\m@th\cdot$}}%
                     \raisebox{-0.3ex}{$\m@th\cdot$}%
                     \rlap{%
                     \raisebox{0.3ex}{$\m@th\cdot$}}
                     \raisebox{-0.3ex}{$\m@th\cdot$}}
                     =}
\makeatother

\newcommand{\dom}{\mathbb{D}}
\newcommand{\Z}{\mathbb{Z}}
\newcommand{\N}{\mathbb{N}}
\newcommand{\RP}{\mathbb{R}_{\geq 0}}
\newcommand{\myprop}{\mathbb{P}}

\newcommand{\mtl}{\textup{MTL}}
 \newcommand{\rmtl}{\textup{SMTL}}
 \newcommand{\smtl}{\rmtl}
 \newcommand{\puremtl}{\textup{pureMTL}}
 \newcommand{\mtlpureuna}{\textup{pureUnaMTL}}
 \newcommand{\mtlmodels}{\models_{\tiny{\mtl}}}
 \newcommand{\ltlmodels}{\models_\ltl}

 \newcommand{\tptlmodels}{\models_\tptl}

\newcommand{\rhs}{\mathsf{rhs}}
\newcommand{\val}{\mathsf{val}}

\newcommand{\true}{\textup{true}}
\newcommand{\false}{\textup{false}}
\newcommand{\sep}{\textit{ }|\textit{ }} 
\newcommand{\U}{\textup{U}} 
\newcommand{\p}{p}
\newcommand{\finally}{\textup{F}}
\newcommand{\glob}{\textup{G}}
\newcommand{\X}{\textup{X}} 
\newcommand{\urk}{\mathsf{Rank}} 
\newcommand{\ev}{\bar{\mathsf{0}}} 
\newcommand{\inft}{\mathsf{Lemn}} 

\newcommand{\tptl}{\textup{TPTL}}
\newcommand{\tptlex}{\textup{TPTL_F}}
\newcommand{\tptleq}{\textup{FreezeLTL}}
\newcommand{\tptlpos}{\textup{TPTL}_+}
\newcommand{\tptlexeq}{\textup{FreezeLTL_{F}}}
\newcommand{\regtptlexeq}[1]{\textup{FreezeLTL_{F}^{#1}}}
\newcommand{\regtptleq}[1]{\textup{FreezeLTL^{#1}}}
\newcommand{\regtptlex}[1]{\textup{TPTL_F^{#1}}}
\newcommand{\freeze}{\textup{FreezeLTL}}
\newcommand{\pure}{\textup{Pure}}

\newcommand{\ltl}{\textup{LTL}}
\newcommand{\ltlex}{\textup{LTL_F}}

 %
\newcommand{\problemx}[3]{
\vspace{3mm}
\par\noindent{\bf#1}\par\nobreak\vskip.2\baselineskip
\begingroup\clubpenalty10000\widowpenalty10000
\setbox0\hbox{\bf INPUT:\ \ }\setbox1\hbox{\bf QUESTION:\ \ }
\dimen0=\wd0\ifnum\wd1>\dimen0\dimen0=\wd1\fi
\vskip-\parskip\noindent
\hbox to\dimen0{\box0\hfil}\hangindent\dimen0\hangafter1\ignorespaces#2\par
\vskip-\parskip\noindent
\hbox to\dimen0{\box1\hfil}\hangindent\dimen0\hangafter1\ignorespaces#3\par
\endgroup\vspace{3mm}
}


\newcommand{\ZZ}{\ensuremath{\mathbb{Z}}\xspace}
\newcommand{\NN}{\ensuremath{\mathbb{N}}\xspace}

\newcommand{\Q}{\mathbb{Q}}
\newcommand{\QP}{\mathbb{Q}_{\geq 0}}
\newcommand{\Rat}{\Q}

\newcommand{\A}{\mathcal{A}}
\newcommand{\Aa}{\ensuremath{\mathcal{A}}\xspace}
\newcommand{\aff}{\mathbb{A}}

\newcommand*\ie{\textit{i.e.}\xspace}
\newcommand*\viceversa{\textit{vice versa}\xspace}

\newcommand{\tm}{\mathcal{T}}
\newcommand{\mtlex}{\textup{MTL_F}}
\newcommand{\mtlpos}{\textup{MTL}_+}

 \newcommand{\tcm}{\mathcal{M}}
 \newcommand{\inst}{\mathcal{I}}
 \newcommand{\newinst}[1]{\texttt{I}_{\texttt{#1}}}
 \newcommand{\counter}[1]{\texttt{C}_{\texttt{#1}}}
 \newcommand{\jumpset}[1]{\texttt{S}_{\texttt{#1}}}
 \newcommand{\zerojumpset}[1]{\texttt{S}_{\texttt{#1}}^{\texttt{0}}}
 \newcommand{\onejumpset}[1]{\texttt{S}_{\texttt{#1}}^{\texttt{1}}}

 \newcommand{\R}{\textup{R}} 
\renewcommand{\P}{P} 

\newcommand{\urnk}{\mathsf{Urk}}
\newcommand{\logic}{\mathbb{L}}
\newcommand{\class}{\mathbb{C}}

\newcommand{\blank}{\qed}
\newcommand{\conf}{\gamma}

\newcommand{\dummy}{\mathsf{dummy}}
\newcommand{\proj}{\mathsf{proj}}

\newcommand{\Sequiv}{\simeq^S}

\newcommand{\NotDefInMTL}{\mathsf{Equalin2Positions}}
\newcommand{\Uweak}{\U_{\mathsf{weak}}}

 \newcommand{\comp}{\mathsf{comp}}

\begin{abstract}
  \noindent 
  Metric temporal logic ($\mtl$) and timed propositional temporal logic ($\tptl$) are quantitative extensions of linear temporal logic, which are prominent and widely used in the verification of real-timed systems. 
  It was recently shown that the path-checking problem for $\mtl$, when evaluated over finite timed words, is in 
  the parallel complexity class $\mync$. 
  In this paper, we derive precise complexity results for the path-checking problem for $\mtl$ and $\tptl$ when evaluated over infinite data words over the non-negative integers. Such words may be seen as the behaviours of one-counter machines. For this setting, we give a complete analysis of the complexity of the path-checking problem depending on the number of register variables and the encoding of constraint numbers  (unary
or binary). 
  As the two main results, we prove that the path-checking problem for $\mtl$ is $\ptime$-complete, 
  whereas the path-checking problem for $\tptl$ is
  $\pspace$-complete. 
  The results yield the precise complexity of model checking deterministic one-counter machines against formulas of $\mtl$ and $\tptl$. 
\end{abstract}

\maketitle


\section{Introduction}
{\em Linear temporal logic} ($\ltl$) is nowadays one of the main logical formalisms for describing system behaviour.
Triggered by real-time applications, various timed extensions of {\sf LTL} have been invented. Two of the most prominent examples are $\mtl$ (metric temporal logic) \cite{K90} and $\tptl$ (timed propositional temporal logic) \cite{DBLP:journals/jacm/AlurH94}. In $\mtl$,  the temporal until modality  ($\U$) is annotated with time intervals.
For instance, the formula $p \, \U_{[2,3)} \, q$ holds at time $t$, if there is a time $t'\in [t+2, t+3)$, where $q$ holds, and $p$ holds during the interval
$[t, t')$. $\tptl$ is a more powerful logic~\cite{DBLP:journals/corr/CarapelleFGQ13} that is equipped with a \emph{freeze formalism}. It uses register variables, which can be set to the current time value and later, these register variables can be compared with the current time value. For instance, the above $\mtl$ formula $p\, \U_{[2,3)}\, q$ is equivalent to the $\tptl$ formula $x.(p \,\U\, (q \wedge 2 \leq x < 3))$. Here, the constraint $2 \leq x < 3$ should be read as: The difference of the current time value and the value stored in $x$ is in the interval $[2,3)$. 
Formulas in $\mtl$ and $\tptl$ are evaluated over finite or infinite \emph{timed words} of the form $(a_0,t_0)(a_1,t_1)\dots$, where $t_0,t_1,\dots$ is a  monotonically increasing sequence of real-valued \emph{timestamps}, and the 
	$a_i$ are actions that take place at timestamp $t_i$. 
	In the context of formal verification of timed automata, 
	\emph{satisfiability} and \emph{model checking} for $\mtl$ and $\tptl$ have been studied intensively in the past, see Table~\ref{table:complexity_results} for a summary of the most important results.

	\begin{table}[t]
		\centering
  \scalebox{1}{
    \def\arraystretch{1.2}
	\begin{tabular}{ l|l|l}
		& {\bf Satisfiability} & {\bf Model Checking}  \\
			\hline
			{\bf Real-Timed Words} &  &   \\
			\hspace{5mm}$\mtl$ finite words & $\mathcal{F}_{\omega^\omega}$-complete ~\cite{DBLP:journals/lmcs/OuaknineW07,DBLP:conf/icalp/SchmitzS11} & $\mathcal{F}_{\omega^\omega}$-complete ~\cite{DBLP:journals/lmcs/OuaknineW07,DBLP:conf/icalp/SchmitzS11}   \\
			\hspace{5mm}$\mtl$ infinite words & undecidable ~\cite{DBLP:conf/fossacs/OuaknineW06} & undecidable ~\cite{DBLP:conf/fossacs/OuaknineW06}  \\
			\hspace{5mm}$\tptl$ & undecidable~\cite{DBLP:journals/iandc/AlurH93}  & undecidable~\cite{DBLP:journals/iandc/AlurH93}  \\
			\hline
			{\bf Discrete Timed Words} &  &   \\
			\hspace{5mm}$\mtl$  & $\expspace$-complete ~\cite{DBLP:journals/iandc/AlurH93} & $\expspace$-complete ~\cite{DBLP:journals/iandc/AlurH93} \\
			\hspace{5mm}$\tptl$  & $\expspace$-complete ~\cite{DBLP:journals/jacm/AlurH94} & $\expspace$-complete ~\cite{DBLP:journals/jacm/AlurH94} \\
			\hline
			{\bf Data Words} & &  \\
			\hspace{5mm}$\mtl$ & undecidable~\cite{DBLP:conf/lata/CarapelleFGQ14} & undecidable~\cite{DBLP:conf/lata/Quaas13}  \\
			\hspace{5mm}$\tptl$  & undecidable~\cite{DBLP:conf/lata/CarapelleFGQ14} & undecidable~\cite{DBLP:conf/lata/Quaas13} \\
			\hspace{5mm}$\freeze$ & undecidable~\cite{DBLP:journals/tocl/DemriL09} & undecidable~\cite{DBLP:journals/tcs/DemriLS10} \\
			\hspace{5mm}$\freeze^1$ finite words & $\mathcal{F}_\omega$-complete~\cite{DBLP:journals/tocl/DemriL09,DBLP:conf/lics/FigueiraFSS11} & undecidable~\cite{DBLP:journals/tcs/DemriLS10} \\
			\hspace{5mm}$\freeze^1$ infinite words & undecidable~\cite{DBLP:journals/tocl/DemriL09} &  undecidable~\cite{DBLP:journals/tcs/DemriLS10} 
\end{tabular}
  }
  \caption[Table caption text]{Complexity results for the satisfiability problem and the model-checking problem for the logics $\mtl$, $\tptl$, $\freeze$, and its one-variable fragment $\freeze^1$. 
		By \emph{real-timed words} (\emph{discrete timed words}, respectively) we mean timed words with  timestamps in the non-negative reals (non-negative integers, respectively). Model checking for timed words is done for timed automata, and model checking for data words is done for one-counter machines. We do not distinguish between finite and infinite words if this does not influence the complexity.}
  \label{table:complexity_results}

\end{table}

	 The freeze mechanism from $\tptl$ has also received attention in connection with \emph{data words}. 
	 Data words generalize timed words and are of the form  $(a_0,d_0)(a_1,d_1)\dots$, where the data values $d_0, d_1, \dots$ come from an arbitrary infinite data domain. 
	 Following~\cite{DBLP:journals/tcs/DemriLS10}, we regard data words as computation paths of \emph{one-counter machines}, \emph{i.e.}, we study data words over the domain of the non-negative integers. Note  that, different to timed words, the sequence of data values $d_0, d_1, \dots$ does not have to be monotonically increasing. 
	 Applications for 
	 data words can be seen in areas where data streams of discrete values have to be analyzed, and the focus is on the dynamic variation of the values (e.g. streams of discrete sensor data or stock charts).

	 For reasoning about data words, besides $\mtl$ and $\tptl$ in~\cite{DBLP:conf/lata/CarapelleFGQ14}, a strict fragment of $\tptl$ called  $\fltl$ is studied  in~\cite{DBLP:journals/tocl/DemriL09,DBLP:journals/tcs/DemriLS10}. With formulas in $\fltl$ one can only test \emph{equality} of the value of a variable $x$ and the data value at the current position in a data word. 
	 Table~\ref{table:complexity_results} shows that for all these logics, model checking one-counter machines and the satisfiability problem are undecidable, with the exception of the satisfiability problem for the one-variable fragment of $\fltl$ when evaluated over finite data words. 
	 The situation changes if one considers model checking of \emph{deterministic} one-counter machines: in this case, $\fltl$-model checking is $\pspace$-complete~\cite{DBLP:journals/tcs/DemriLS10}. 
	 Note that model checking of deterministic one-counter machines is a special case of the following \emph{path-checking problem}: given a data word $w$ and some formula $\varphi$, does $w$ satisfy $\varphi$?

	 The path-checking problem plays a key role in 
	 \emph{run-time verification}~\cite{DBLP:journals/entcs/FinkbeinerS01,DBLP:conf/concur/MarkeyS03,DBLP:conf/formats/MalerN04}, specifically in \emph{offline monitoring}, where  the satisfaction of a specification formula is tested only for an individual single computation path of the observed system. 
	 Run-time verification may be the only practical alternative to check that certain temporal properties hold, in situations where the source code of the system under consideration is not available, or where model checking of the system is unfeasible if not undecidable, as it is the case for  one-counter machines.  

 For $\ltl$ and periodic words without data values, it was shown in~\cite{DBLP:conf/icalp/KuhtzF09} that the path-checking problem can be solved using an efficient parallel algorithm. More precisely, the problem belongs to $\textup{AC}^1(\textup{LogDCFL})$, a subclass of $\mync$. This result solved a long standing open problem; the best known lower bound is $\mync^1$, arising from $\mync^1$-hardness of evaluating Boolean expressions. 
 The $\textup{AC}^1(\textup{LogDCFL})$-upper complexity bound was later even established for the path-checking problem for $\mtl$ over finite timed words~\cite{DBLP:conf/icalp/BundalaO14}.

 In this paper, we continue the study of the path-checking problem started in~\cite{DBLP:journals/tcs/DemriLS10} for $\tptl$ and data words over the non-negative integers, \ie, we regard data words as the behaviours of one-counter machines.  
 Note that $\tptl$ is strictly more expressive than $\fltl$ in that, in contrast to the latter, with $\tptl$ one can express that the difference of the current data value in a data word and the value stored in a register belongs to a certain interval. 
 We further investigate path checking for $\mtl$, because, as it was recently proved~\cite{DBLP:journals/corr/CarapelleFGQ13}, $\mtl$ is strictly less expressive than $\tptl$ in the setting of data words over the non-negative integers. 
 More specifically, we investigate the path-checking problems for $\tptl$ over
data words that can be either finite or infinite periodic; in the latter case the
 data word is specified by an initial part, a period, and an offset number, which is added to the data values
 in the period after each repetition of the period.


We show that the $\textup{AC}^1(\textup{LogDCFL})$-membership result of \cite{DBLP:conf/icalp/BundalaO14} for path checking of $\mtl$ over finite timed words is quite sharp in the following sense: 
path checking for
$\mtl$ over (finite or infinite) data words as well as path checking for the one-variable fragment of $\tptl$ evaluated over monotonic
(finite or infinite) data words is $\ptime$-complete. 
Moreover, path checking for $\tptl$ (with an arbitrary number of register variables)
over finite as well as infinite periodic data words becomes
$\pspace$-complete. We also show that $\pspace$-hardness already holds (i) for the fragment of
$\tptl$ with only two register variables and (ii) for full $\tptl$, where all interval borders are encoded in unary (the latter result can be shown by a straightforward adaptation of the $\pspace$-hardness proof in \cite{DBLP:journals/tcs/DemriLS10}). These results yield a rather complete picture on the complexity
 of path checking for $\mtl$ and $\tptl$, see Figure~\ref{fig-paths} at the end of the paper.
We also show that $\pspace$-membership for the path-checking problem for $\tptl$ still holds
if data words are specified succinctly by so called straight-line programs, which have the ability to generate
data words of exponential length.  For ordinary $\ltl$ it was shown in \cite{DBLP:conf/concur/MarkeyS03} 
that path checking is $\pspace$-complete
if paths are represented by straight-line programs. Our result extends the $\pspace$ upper bound from \cite{DBLP:conf/concur/MarkeyS03}.

Since the infinite data word produced by a deterministic one-counter machines is periodic, we can transfer all complexity results for the
infinite periodic case to deterministic one-counter machines, assuming that update numbers are encoded in unary notation.  
 The third author proved recently
 that model checking for $\tptl$ over deterministic one-counter machines is decidable
 \cite{DBLP:conf/lata/Quaas13}, but the complexity remained open. Our results show that
 the precise complexity is $\pspace$-complete.  This also generalizes the $\pspace$-completeness result for
 $\fltl$ over deterministic one-counter machines \cite{DBLP:journals/tcs/DemriLS10}. Our $\pspace$ upper bound for $\tptl$ model 
 checking over deterministic one-counter machines also holds when the update numbers of the one-counter machines are given in binary notation.
These are called {\em succinct one-counter automata} in \cite{DBLP:conf/icalp/GollerHOW10,DBLP:conf/concur/HaaseKOW09}, where the complexity of reachability problems 
and model-checking problems for various temporal logics over succinct nondeterministic one-counter automata were studied.

An extended abstract of this paper without full proofs 
appeared in \cite{FLQ15}.

\section{Temporal Logics over Data Words}

\paragraph{\bf Data Words}
Let $\myprop$ be a finite set of {\em atomic propositions}. A \emph{word} over $\myprop$ is a finite or infinite sequence $P_0 P_1 P_2 \dots$, where $P_i\subseteq \myprop$ for all $i\in\N$. 
If $P_i=\{p_i\}$ is a singleton set for every position $i$ in the word, 
then we may also write $p_0 p_1 p_2 \dots$.   
 A \emph{data word} over $\myprop$ is a finite or infinite sequence
 $(P_0,d_0)(P_1,d_1)(P_2,d_2)\dots$, where  $(P_i,d_i)\in (2^\myprop\times\N)$ for all $i\in\N$.
A data word is {\em monotonic} ({\em strictly monotonic}), if $d_i \leq d_{i+1}$ ($d_i < d_{i+1}$) for all $i\in\N$. 
A {\em pure} data word is a finite or infinite sequence $d_0 d_1 d_2 \dots$ of natural numbers; it can be identified
with the data word  $(\emptyset,d_0)(\emptyset,d_1) (\emptyset,d_2) \dots$.
        We use $(2^\myprop \times\N)^*$ and $(2^\myprop\times\N)^\omega$, respectively, to denote the
        set of finite and infinite, respectively, data words over $\myprop$. 
        We use $|w|$ to denote the \emph{length} of a data word $w$, \ie, the number of all pairs $(P_i,d_i)$ in $w$. If $w$ is infinite, then $|w|=+\infty$.

        Let $w=(P_0,d_0)(P_1,d_1)\dots $ be a data word, and let $i\in\{0,\dots,|w|\}$ be a position in $w$. 
        We define $w[{i\!:}]$ to be the suffix of $w$ starting in position $i$, \ie, ${w[{i\!:}]\defeq(P_i,d_i)(P_{i+1},d_{i+1})\dots}$.
        For an integer $k\in\Z$ satisfying  $d_i+k \geq 0$ for all $i \geq 0$,  we define the data word
        $w_{+k}\defeq(P_0,d_0+k)(P_1,d_1+k)\dots$.        
        
        We use $u_1u_2$ to denote the \emph{concatenation} of two data words $u_1$ and $u_2$, where $u_1$ has to be finite. 
        For finite data words $u_1,u_2$ and $k\in\N$, we define
        $$u_1(u_2)^\omega_{+k} \defeq u_1 u_2 (u_2)_{+k} (u_2)_{+2k} (u_2)_{+3k}\dots$$
        
        For complexity considerations, the encoding of the data values and the offset number $k$ (in an infinite data word) makes a difference.
        We speak of {\em unary} (resp., {\em binary}) encoded data words if all these numbers are given in unary (resp., binary) encoding.

        \paragraph{\bf Linear Temporal Logic}
        Given a finite set $\myprop$ of propositions, the set of formulas of linear temporal logic ($\ltl$, for short) is built up from $\myprop$ by Boolean connectives and the {\em until} modality $\U$ using the following grammar:
        \[
        \varphi \ndefeq \true\sep p \sep \neg\varphi \sep \varphi\wedge\varphi \sep \varphi\U\varphi
        \]
        where $p\in\myprop$.
        Formulas of $\ltl$ are interpreted over \emph{words} over $\myprop$. Let $w=P_0\,P_1\,P_2\dots$ be a word over $\myprop$,
        and let $i$ be a position in $w$. We define the {\em satisfaction relation for $\ltl$} inductively as follows:
        \begin{itemize}
        \item $(w,i) \models_\ltl \true$.
        \item $(w,i) \models_\ltl p$ if, and only if,   $p \in P_i$.
            \item $(w,i) \models_\ltl \neg\varphi$ if, and only if,  $(w,i)\not\models_\ltl \varphi$.
            \item $(w,i) \models_\ltl \varphi_1\wedge\varphi_2$ if, and only if,    $(w,i)\models_\ltl \varphi_1$ and $(w,i)\models_\ltl \varphi_2$.
            \item $(w,i)\models_\ltl \varphi_1\U\varphi_2$ if, and only if,  there exists a position $j>i$ in $w$ such that $(w,j)\models_\ltl\varphi_2$, and $(w,k)\models_\ltl\varphi_1$ for all positions $k$ with $i<k<j$.
        \end{itemize}
        We say that a word \emph{satisfies} an $\ltl$ formula $\varphi$, written $w\models\varphi$, if $(w,0)\models\varphi$.
        
        We use the following standard abbreviations:
        \begin{align*}
        	\false&\defeq\neg\true &  \F\varphi&\defeq\true \U\varphi\\
            \varphi_1\vee\varphi_2 &\defeq\neg(\neg\varphi_1\wedge\neg\varphi_2) &
            \G\varphi& \defeq\neg\F\neg\varphi\\
            \varphi_1\rightarrow \varphi_2 &\defeq \neg\varphi_1\vee\varphi_2 &
            \X\varphi &\defeq \false \U\,\varphi\\
            \varphi_1\,\R\,\varphi_2&\defeq \neg(\neg\varphi_1\U\neg\varphi_2) &
            \X^{m}\varphi &\defeq \underbrace{\X\dots\X}_{m}\varphi
        \end{align*}

The modalities $\X$ (\emph{next}), $\F$ (\emph{finally}) and $\G$ (\emph{globally}), respectively, are \emph{unary} operators, which refer to the \emph{next} position, \emph{some} position in the future and \emph{all} positions in the future, respectively. The binary modality $\R$ is the \emph{release} operator, which is useful to transform a formula into an equivalent \emph{negation normal form}, where the negation operator ($\neg$) may only  be applied to $\true$ or to propositions.

\paragraph{\bf Metric Temporal Logic and Timed Propositional Temporal Logic}
Metric Temporal Logic, $\mtl$ for short, extends $\ltl$ in that the until modality $\U$ may be annotated with an interval over $\Z$. More precisely, the set of $\mtl$ formulas is defined by the following grammar:
$$\varphi \ndefeq \true\sep p \sep \neg\varphi \sep \varphi\wedge\varphi \sep
\varphi\U_{I}\varphi $$
where $p\in\myprop$ and $I \subseteq \Z$ is an interval with endpoints in
$\Z\cup\{-\infty,+\infty\}$.

Formulas in $\mtl$ are interpreted over data words. 
Let $w=(P_0,d_0)(P_1,d_1)\dots$ be a data word over $\myprop$, and let $i$ be a position in $w$. 
We define the {\em satisfaction relation for $\mtl$}, denoted by $\mtlmodels$,
inductively as follows (we omit the obvious cases for $\neg$ and $\wedge$):

 \begin{itemize}
 \item $(w,i) \mtlmodels p$ if, and only if,   $p\in P_i$.
 \item $(w,i)\mtlmodels \varphi_1\U_{I}\varphi_2$ if, and only if,  there exists a position $j>i$ in $w$ such that $(w,j)\mtlmodels\varphi_2$, $d_j - d_i \in I$, and  $(w,k)\mtlmodels\varphi_1$ for all positions $k$ with $i<k<j$.
        \end{itemize}
We say that a data word \emph{satisfies} an $\mtl$ formula $\varphi$, written
$w\mtlmodels\varphi$, if $(w,0)\mtlmodels\varphi$.
We use the same syntactic abbreviations as for $\ltl$, where every temporal operator is annotated with an interval, \eg, 
$\F_I\varphi\defeq\true\U_I\varphi$ and $\X_I\varphi \defeq \false \U_I\varphi$. For annotating the temporal operators, we may also use pseudo-arithmetic expressions of the form $x\sim c$, where ${\sim} \in\{<,\le,=,\ge,>\}$ and $c\in\Z$. For instance, we may write $\F_{=2} p$ as abbreviation for $\F_{[2,2]}p$. 
        If $I=\Z$, then we may omit the annotation $I$ on $\U_I$.

        Some of our results for lower bounds already hold for fragments of $\mtl$, which we explain in the following. 
        We write $\mtl(\F,\X)$ to denote the \emph{unary} fragment of $\mtl$ in which the only temporal modalities allowed are $\F$ and $\X$, and we write
        $\mtl(\F)$ to denote the fragment of $\mtl$ in which $\F$ is the only allowed temporal modality. 
        We write $\pure\mtl$ ($\pure\mtl(\F,\X)$, respectively) to denote the set of $\mtl$ formulas ($\mtl(\F,\X)$ formulas, respectively), in which no propositional variable from $\myprop$ is used.

%


Next, we define Timed Propositional Temporal Logic,  $\tptl$ for short. 
 Let $V$ be a countable set of \emph{register variables}. Formulas of $\tptl$ are built by the following grammar:
        \[
        \varphi \ndefeq \true \sep p \sep x\sim c \sep \neg\varphi \sep \varphi\wedge\varphi
        \sep \varphi\U\varphi \sep x.\varphi
        \]
        where $p\in\myprop$, $x\in V$, $c\in\Z$, and $\sim\,\in\{<,\leq,=,\geq, >\}$. 

        A \emph{register valuation} $\nu$ is a function from $V$ to $\Z$.
        Given a register valuation $\nu$, a data value $d\in\Z$, and a variable $x\in V$,
        we define the register valuations $\nu+d$ and $\nu[x\mapsto d]$ as follows:
        $(\nu+d)(y)=\nu(y)+d$ for every $y\in V$,
        $(\nu[x \mapsto d])(y)=\nu(y)$ for every $y\in V\backslash\{x\}$, and $(\nu[x \mapsto d])(x)=d$.
        Let $w=(P_0,d_0)(P_1,d_1)\dots$ be a data word over $\myprop$, let $\nu$ be a register
        valuation, and let $i$ be a position in $w$.
        The satisfaction relation for $\tptl$, denoted by $\tptlmodels$, is defined inductively in the obvious way; we only give the definitions for the new formulas:
        \begin{itemize}
            \item $(w,i,\nu)\tptlmodels x\sim c$ if, and only if,  $d_i -\nu(x)\sim c$.
            \item $(w,i,\nu) \tptlmodels x.\varphi$ if, and only if,  $(w,i,\nu[x\mapsto d_i])\tptlmodels\varphi$.
        \end{itemize}

        Intuitively, $x. \varphi$, means that  $x$ is \emph{updated} to the data value at the current position of the data word, and $x \sim c$ means that, compared to the last time that $x$ was updated, the data value has changed by at least (at most, or exactly, respectively) by $c$.
        We say that a data word $w$ satisfies a $\tptl$ formula $\varphi$, written
        $w\models\varphi$, if $(w,0,\mathbf{d_0})\tptlmodels\varphi$, where $\mathbf{d_0}$ denotes the valuation that maps all register variables to the initial data value $d_0$ of the data word $w$.

        We use the same syntactic abbreviations as for $\ltl$.
        We define the fragments $\tptl(\F,\X)$, $\tptl(\F)$, and $\pure\tptl$ like the
        corresponding fragments of $\mtl$.
Additionally, we define $\freeze$ to be the subset of $\tptl$ formulas $\varphi$
where every subformula $x\sim c$ of $\varphi$ must be of the form $x=0$. 
Further, for every $r\in\N$, we use $\tptl^r$ to denote the fragment of $\tptl$ in which at most $r$ different register variables occur; similarly for other $\tptl$-fragments.

 For complexity considerations, it makes a difference whether the numbers $c$ in constraints $x \sim c$
     are binary or unary encoded, and similarly for the interval borders in $\mtl$.
     We annotate our logics $\logic$ by an index $\textup{u}$ or $\textup{b}$, \ie, we write $\logic_{\textup{u}}$ respectively $\logic_{\textup{b}}$, to emphasize that
      numbers are encoded in unary (resp., binary) notation.
%
    The \emph{length} of a ($\tptl$ or $\mtl$) formula $\psi$, denoted by $\vert\psi\vert$, is the number of symbols occurring in $\psi$.

    In the rest of the paper, we study the path-checking problems for our logics over data words.
    Data words can be (i) finite or infinite, (ii) monotonic or non-monotonic, (iii) pure or non-pure,
    and (iv) unary encoded or binary encoded.
    For a logic $\logic$ and a class of data words $\class$, we consider the {\em path-checking problem  for $\logic$ over $\class$}: 
    given some data word $w \in \class$ and some formula $\varphi \in \logic$, does $w\models_\logic \varphi$ hold?

\section{Background from Complexity Theory}

We assume that the reader is familiar with the complexity classes $\ptime$ (deterministic 
polynomial time) and $\pspace$ (polynomial space), more background can be found for instance
in \cite{AroBar09}. Recall that, by Savitch's theorem, nondeterministic polynomial space is equal to 
deterministic polynomial space. All completeness results in this paper refer to logspace reductions.

We will make use of well known characterizations of $\ptime$ and $\pspace$  in terms of alternating Turing machines.
An {\em alternating Turing machine} is a nondeterministic Turing machine, whose state set $Q$ is partitioned in four
disjoint sets $Q_{\text{acc}}$ (accepting states), $Q_{\text{rej}}$ (rejecting states),  $Q_\exists$ (existential states), and $Q_\forall$
(universal states). A configuration $c$, where the current state is
$q$, is {\em accepting} if (i) $q\in Q_{\text{acc}}$ or
(ii) $q\in Q_\exists$ and there exists an accepting successor
configuration of $c$ or
(iii) $q\in Q_\forall$ and all successor configurations of
$c$ are accepting. The machine accepts an input $w$ if, and only
if, the initial configuration for $w$ is accepting.
It is well known that the class of all languages that can be accepted
by an alternating Turing machine in polynomial time ($\aptime$) is equal
to $\pspace$, and that the class of all languages that can be accepted
by an alternating Turing machine in logarithmic space ($\alogspace$) is equal
to $\ptime$.

A couple of times we will mention the complexity class $\textup{AC}^1(\textup{LogDCFL})$. 
For completeness, we present the definition: the class $\textup{LogDCFL}$ is the class of all languages that are logspace
reducible to a deterministic context-free language. Then $\textup{AC}^1(\textup{LogDCFL})$
is the class of all problems that can be solved by a logspace-uniform 
circuit family of polynomial size and logarithmic depth, where in addition to ordinary Boolean
gates (NOT, AND, OR) also oracle gates for problems from $\textup{LogDCFL}$ 
can be used. More details can be found in  \cite{DBLP:conf/icalp/KuhtzF09}.
The class  $\textup{AC}^1(\textup{LogDCFL})$ belongs to $\mync$ (the class of all problems
	that can be solved on a parallel random-access machine (PRAM) in polylogarithmic time with polynomially many processors),
which, in turn, is contained in $\ptime$.

\section{Upper Complexity Bounds}
\label{sec-upper-bounds}

In this section we prove the upper complexity bounds for the path-checking problems. We distinguish between (i) unary or binary encoded data words, and (ii) finite and infinite data words.

\subsection{Polynomial Space Upper Bound for $\tptl$} 
For the most general path-checking problem ($\tptl_\textup{b}$ over infinite binary encoded data words) we can devise an alternating
polynomial time (and hence a polynomial space) algorithm by constructing an alternating Turing machine that, given a $\tptl_{\textup{b}}$ formula $\varphi$ and an infinite binary encoded data word $w$, has an accepting run if, and only if, $w\tptlmodels\varphi$. 
 The main technical difficulty is to bound the position
in the infinite data word and the values of the register valuation, so that they can be stored in polynomial space. 
 \begin{theorem}\label{Path check TPTL-upper-PSPACE}
 	 Path checking for $\tptl_\textup{b}$ over infinite binary encoded data words is in $\pspace$.
    \end{theorem}
   Before we give the proof of Theorem~\ref{Path check TPTL-upper-PSPACE}, 
    we introduce some helpful notions and prove some lemmas.

    \paragraph{\bf Relative semantics.}
Let $w$ be a data word, $i\in\N$ be a position in $w$, and let $\delta$ be a register valuation.
For technical reasons, we introduce a \emph{relative satisfaction relation} for $\tptl$, denoted by $\rmodels$, as  follows.
For Boolean formulas, $\rmodels$ is defined like $\models_\tptl$. For the other operators we define:
\begin{itemize}
\item $(w,i,\delta)\rmodels\varphi_{1}\U\varphi_{2}$ if, and only if, there exists a position $j>i$ in $w$ such that 
	$(w,j,\delta+(d_j-d_i))\rmodels\varphi_{2}$, and $(w,k,\delta+(d_t-d_i))\rmodels\varphi_{1}$ for all positions $k$ with $i<k<j$
          \item $(w,i,\delta)\rmodels x\sim c$ if, and only if, $\delta(x)\sim c$
          \item  $(w,i,\delta)\rmodels x.\varphi$ if, and only if, $(w,i,\delta[x\mapsto 0])\rmodels\varphi$.
        \end{itemize}
        We say that the data word $w$ satisfies the formula $\varphi$ under the relative semantics, written $w\rmodels\varphi$, if $(w,0,\mathbf{0})\rmodels\varphi$,
where $\mathbf{0}$ denotes the valuation function that maps all register variables to $0$.

The main advantage of the relative semantics is the following: under the normal $\tptl$ semantics, a constraint $x \sim c$ is true under a valuation $\nu$ at a position with data value
$d$, if $d-\nu(x) \sim c$ holds. In contrast, under the relative semantics,  a constraint $x \sim c$ is true under a valuation $\delta$, if $\delta(x) \sim c$ holds,
\ie, the data value at the current position is not important.
The following lemma implies that $w\tptlmodels\varphi$ if, and only if, $w\rmodels\varphi$, which allows us to work with the relative semantics. 
    \begin{lemma}\label{equivalence of semantics}
    	    Let $w$ be a data word and $d_i$ the data value at position $i$.
	    If $\delta(x)=d_{i}-\nu(x)$ for every register variable $x$, then for every $\tptl$ formula $\varphi$,  $(w,i,\nu)\tptlmodels \varphi$ if, and only if, $(w,i,\delta)\rmodels \varphi$. 
    \end{lemma}
    \begin{proof}
    	   The proof is by induction on the structure of the formula $\varphi$. We only consider the non-trivial cases.    	    
    	    
    	    \noindent
    	    {\bf Base case.} Assume $\varphi=x\sim c$.
    	    \begin{eqnarray*}
    	    	    (w,i,\nu)\tptlmodels \varphi & \Leftrightarrow  & d_i-\nu(x)\sim c \quad\mbox{ (definition $\tptlmodels$)}\\
    	    	    & \Leftrightarrow & d_i-d_i+\delta(x)\sim c \quad\mbox{ (assumption)}\\
    	    	    & \Leftrightarrow & (w,i,\delta)\rmodels\varphi \quad\mbox{ (definition $\rmodels$)} \, .
\end{eqnarray*}

\noindent
{\bf Induction step.}
    \begin{itemize}    	    	
\item Assume $\varphi=\varphi_1 \U \varphi_2$. We have	
	\[\begin{aligned} 
		& \, (w,i,\nu)\tptlmodels\varphi_{1}\U\varphi_{2} \\
		 \Leftrightarrow  \,  & \,  \mbox{there exists } i<j<|w|. (w,j,\nu)\tptlmodels\varphi_2, \mbox{ and} \\
		& \,  (w,t,\nu)\tptlmodels\varphi_1 \mbox{ for all } i<t<j \quad\mbox{ (definition $\tptlmodels$)}\\
		\Leftrightarrow \,  & \,  \mbox{there exists } i<j<|w|. (w,j,\delta+ (d_j-d_i))\rmodels\varphi_2, \mbox{ and}\\
		& \, (w,t,\delta+(d_t-d_i))\rmodels\varphi_1 \mbox{ for all } i<t<j \quad\mbox{(induction hypothesis)} \\ 
		\Leftrightarrow \, & \,  (w,i,\delta)\rmodels\varphi_{1}\U\varphi_{2} \quad\mbox{(definition $\rmodels$)}
        	\end{aligned}\]
    	    
    \item Assume $\varphi=x.\varphi_1$. By definition 
	\[(w,i,\nu)\tptlmodels x.\varphi_{1} \Leftrightarrow (w,i,\nu[x\mapsto d_{i}])\tptlmodels\varphi_{1},\] and  
	\[(w,i,\delta)\rmodels x.\varphi_1 \Leftrightarrow (w,i,\delta[x\mapsto 0])\rmodels \varphi_1.\]
	Clearly, 
	\[(\delta[x\mapsto 0])(x) = 0 = d_i-d_i = d_i - (\nu[x\mapsto d_i])(x),\]
	and 
	\[(\delta[x\mapsto 0])(y)=\delta(y)= d_i - \nu(y)=d_i- (\nu[x\mapsto d_i])(y)\] for every register variable $y\neq x$.
	The result follows by  induction hypothesis on $\delta[x\mapsto 0]$ and $\nu[x\mapsto d_i]$. 
\qedhere
    \end{itemize}
    \end{proof}

For the next three lemmas, we 
let $w = u_{1}(u_{2})^{\omega}_{+k}$, where $u_1$ and $u_2$ are  finite data words,  
and $k \geq 0$.
Further assume $i \geq |u_1|$, and let $\psi$ be a $\tptl$ formula.
\begin{lemma}\label{lemma_claim1}
	For all register valuations $\delta$, ${(w,i,\delta)\rmodels \psi}$ if and
  only if
  ${(w,i+|u_2|,\delta)\rmodels \psi}$.
\end{lemma}

  \begin{proof}
        	 Let $\delta$ be a register valuation. 
  	  Define $\nu=d_i-\delta$ and $\nu'=\nu+k$. 
  	  Lemma~\ref{equivalence of semantics} yields
  	  \[(w,i,\delta)\rmodels \psi \, \Leftrightarrow \, (w,i,\nu)\tptlmodels \psi,\]
  	  and, together with $d_{i+|u_2|}=d_i+k$, 
  	  \[(w,i+|u_2|,\delta)\rmodels \psi  \, \Leftrightarrow \,(w,i+|u_2|,\nu')\tptlmodels \psi.\]
        	We prove that
        	$(w,i,\nu)\tptlmodels \psi \, \Leftrightarrow \, (w,i+|u_2|,\nu')\tptlmodels \psi$;
        	the claim then follows.
        	For $l\ge 0$, let $w_{\ge l}$ denote the suffix of $w$ starting in position $l$. Since $i\ge |u_1|$, we have 
        	$w_{\ge (i+|u_2|)}=w_{\ge i}+k$. Hence, we only need to show that $(w_{\ge i},0,\nu)\tptlmodels\psi \, \Leftrightarrow \, ((w_{\ge i})_{+k},0,\nu')\tptlmodels\psi$. This can be shown by a simple induction on the structure of $\psi$.
        \end{proof}

 \begin{lemma}	\label{lemma_aux}
        Let $\delta_1,\delta_2$ be two register valuations.
        	If for every $j\geq i$ and every subformula $x\sim c$ of $\psi$ we have
        	\begin{equation*} 
        (w,j,\delta_1+(d_j-d_i))\rmodels x\sim c  \ \Leftrightarrow \ (w,j,\delta_2+(d_j-d_i))\rmodels x\sim c,
        \end{equation*}
then we also have
	$(w,i,\delta_1)\rmodels\phi \ \Leftrightarrow \ (w,i,\delta_2)\rmodels\phi.$
        \end{lemma}
        \begin{proof}
        	The proof is by induction on the structure of $\psi$. We only prove the non-trivial cases. 
        	
        	\noindent
        	{\bf Base case.}
        	Let $\psi=x\sim c$. 
        	Clearly, $\delta_k= \delta_k + (d_i-d_i)$ for $k=1,2$. 
        	Hence 
        	\[\begin{aligned} 
	(w,i,\delta_1)\rmodels x\sim c \, &\Leftrightarrow \, (w,i,\delta_1+(d_i-d_i))\rmodels x\sim c \,  \\
	&  \Leftrightarrow \, 
        		(w,i,\delta_2+(d_i-d_i))\rmodels x\sim c \, \Leftrightarrow \, 
        	(w,i,\delta_2)\rmodels x\sim c.
        	\end{aligned}\] 
        	
\noindent
        	{\bf Induction step.}
        	Assume $\psi=x.\psi_1$. 
        	By definition, \[(w,i,\delta_k)\rmodels x.\psi_1 \, \Leftrightarrow  \, (w,i,\delta_k[x\mapsto 0])\rmodels\psi_1\] for $k=1,2$. 
        	Note that the valuations $\delta_1[x\mapsto 0]$
        	and $\delta_2[x\mapsto 0]$ satisfy the premise of the lemma. The result follows by induction hypothesis. 	
        \end{proof}

        \noindent
        For a $\tptl$ formula $\varphi$ and a finite data word $v$ we define:
\begin{eqnarray}
        C_\varphi &=& \max \{ c\in \Z \mid  x\sim c \textrm{ is a subformula of }\varphi\} \label{C}\\
        M_v  &=& \max\{d_i - d_j \mid \textrm{$d_i$ and $d_j$ are data values in } v\} \geq 0 \label{M}
       \end{eqnarray}
We may always assume that $C_\varphi \geq 0$ (we can add a dummy constraint $x \geq 0$).
        Note that in the infinite data word $v_{+k}^\omega$, for all positions $i < j$ we have
        $d_j - d_i + M_v \geq 0$, where $d_i$ and $d_j$ denote the data values  at positions $i$ and $j$, respectively.
          \begin{lemma}
        	\label{lemma_claim2}
        	Let $\delta$ be a register valuation and define $\delta'$ by  ${\delta'(x)=\min\{\delta(x),C_\psi+M_{u_2}+1\}}$ for all $x$.
	For every  subformula $\theta$ of $\psi$, we have
        	$(w,i,\delta)\rmodels\theta$ if, and only if,  $(w,i,\delta')\rmodels\theta$.
        \end{lemma}

        \begin{proof}
 We prove that the premise of Lemma~\ref{lemma_aux} holds for $\delta_1 = \delta$ and $\delta_2 = \delta'$.
 The claim then follows from Lemma~\ref{lemma_aux}. So let $j\geq i$, and let $x\sim c$ be a subformula of $\psi$.
 If $\delta(x)\leq C+M_{u_2}+1$, and hence $\delta'(x)=\delta(x)$, then it is clear that the premise of Lemma~\ref{lemma_aux} is satisfied.
 So assume $\delta(x)>C+M_{u_2}+1$, and hence $\delta'(x)=C+M_{u_2}+1$.
 Then
 \[\delta(x)+d_j - d_i> C+M_{u_2}+1+d_j - d_i\geq C+1\] and \[\delta'(x)+d_j - d_i=C+M_{u_2}+1+d_j - d_i\geq C+1.\]
        Thus, the premise of Lemma~\ref{lemma_aux} holds.
 \end{proof}

 \noindent
 {\em Proof of Theorem~\ref{Path check TPTL-upper-PSPACE}.}
        Fix two finite data words $u_1,u_2$, a  number $k\in \N$ and a $\tptl$ formula $\psi$, and let $w=u_{1}(u_{2})^{\omega}_{+k}$.
        We show that one can decide in  $\aptime=\pspace$ whether $w \models_\tptl \psi$ holds.
        We first deal with the case $k>0$ and later sketch the necessary adaptations for the (simpler) case
        $k=0$. Without loss of generality, we further assume $\psi$ to be in negation normal form.
        Define $C \defeq C_\psi$ and $M \defeq M_{u_2}$ by
        \eqref{C} and \eqref{M}.

        The non-trivial cases in our alternating polynomial time algorithm are  the ones for $\psi = \varphi_{1}\U\varphi_{2}$
        and $\psi = \varphi_{1}\R\varphi_{2}$.
        Consider a position $i$ and a register valuation $\delta$.
        We have $(w,i,\delta)\rmodels\varphi_{1}\U\varphi_{2}$ if, and only if, 
        there exists some position $j>i$ in $w$ such that $(w,j,\delta+(d_j-d_i)\rmodels\varphi_2$, and 
        for all $t$ with $i< t<j$ we have $(w,t,\delta + d_{t}-d_{i})\rmodels\varphi_{1}$.
        Because $w$ is an infinite word, $j$ could be arbitrarily large. Our first goal is to derive a bound on $j$.
        Suppose that $0\leq i\leq|u_{1}|+|u_{2}|-1$; this is no restriction by Lemma~\ref{lemma_claim1}.
        Define
        \begin{eqnarray}
        m_\delta &=& \min\{\delta(x)\mid x\text{ is a register variable in }\psi\}, \label{eq-mdelta} \\
        m_{1} &=& \max\{d_i - d_j \mid \textrm{$d_i$ and $d_j$ are data values in } u_1u_2\}  \text{ and } \label{eq-m1} \\
        m_{2} &=& \min\{d \mid d \text{ is a data value in }u_{2}\}  \label{eq-m2}.
        \end{eqnarray}
        Let $n_\delta \geq2$ be the minimal number such that $m_\delta+m_{2}+(n_\delta-1)k-d_{i}\geq C+M+1$,
        \ie, (here we assume $k > 0$),
        \begin{equation} \label{eq-n}
        n_\delta=\max\{2,\left\lceil \frac{C+M+1+d_{i}-m_\delta-m_{2}}{k}\right\rceil +1\}.
        \end{equation}
        If $h\geq|u_{1}|+(n_\delta-1)|u_{2}|$, then for every register variable $x$ from $\psi$ we have
        $$
        \delta(x)+d_{h}-d_{i} \geq m_\delta + d_h - d_i \geq m_\delta+m_{2}+(n_\delta-1)k-d_{i}\geq C+M+1 .
        $$
        By Lemmas~\ref{lemma_claim1} and \ref{lemma_claim2}, for every $h\geq|u_{1}|+(n_\delta-1)|u_{2}|$ we have
        $$(w,h,\delta+d_{h}-d_{i})\rmodels\varphi_{2} \Leftrightarrow (w,h+|u_{2}|,\delta + d_{h+|u_{2}|}-d_{i})\rmodels\varphi_{2}.$$
        Therefore, the position $j$ witnessing $(w,j,\delta + d_{j}-d_{i})\rmodels\varphi_{2}$
        can be bounded by $|u_{1}|+n_\delta |u_{2}|$. Similarly, we get the same result for $\varphi_{1}\R\varphi_{2}$.

        We sketch an alternating Turing machine $\tm$ that, given a $\tptl_{\textup{b}}$ formula $\psi$ and a data word $w$, has an accepting run if, and only if,  $w\models_\tptl\psi$.
        The machine $\tm$ first computes and stores the value $C+M+1$. In every configuration, $\tm$ stores a triple
        $(i,\delta, \varphi)$, where $i$ is a position in the data word, $\delta$ is a register valuation (with respect to the relative semantics), and $\varphi$ is a subformula of $\psi$.
        By Lemma~\ref{lemma_claim1}, we can restrict $i$ to the interval $[0,|u_{1}|+|u_{2}|)$, and by Lemma~\ref{lemma_claim2}, we can restrict the range of $\delta$ to the interval
        $[-m_1, \max\{m_1, C+M+1\}]$. The machine
        $\tm$ starts with the triple $(0,\mathbf{0},\psi)$, where $\mathbf{0}(x)=0$ for each register variable $x$.
        Then, $\tm$ branches according to the following rules, where we define the function $\rho : \mathbb{N} \to [0, |u_{1}|+|u_{2}|)$ by
        $\rho(z) = z$ for $z  < |u_{1}|$ and $\rho(z) = ((z-|u_{1}|)\text{ mod }|u_{2}|)+|u_{1}|$ otherwise.

\medskip\noindent
  If $\varphi$ is of the form $p$, $\neg p$, or $x\sim c$, then accept if $(w,i,\delta)\rmodels\varphi$, and reject otherwise.

\medskip\noindent
  If $\varphi=\varphi_{1}\wedge\varphi_{2}$, then branch universally to $(i,\delta,\varphi_{1})$ and $(i,\delta,\varphi_{2})$.

\medskip\noindent
  If $\varphi=\varphi_{1}\vee\varphi_{2}$, then branch existentially to $(i,\delta,\varphi_{1})$ and $(i,\delta,\varphi_{2})$.

\medskip\noindent
  If $\varphi=x.\varphi_1$, then go to $(i,\delta[x \mapsto 0],\varphi_1)$.

\medskip\noindent
 If $\varphi=\varphi_{1}\U \varphi_{2}$, then branch existentially to  the following two alternatives.
                \begin{itemize}
                \item Go to $(i, \delta, \varphi)$.
                \item  Compute the value $n_\delta$ according to \eqref{eq-mdelta}, \eqref{eq-m2}, and \eqref{eq-n}, then
                          branch existentially to each value $j \in (i+1,  |u_{1}|+n_\delta|u_{2}|]$, and finally
                          branch universally to each triple from $\{(\rho(t),\delta_t,\varphi_{1})\mid i< t<j\} \cup \{(\rho(j),\delta_j,\varphi_{2})\}$, where for all $x$:
                            \[
                            \begin{array}{l}
                            \delta_j(x) =\begin{cases}
                                \min\{\delta(x)+d_{j}-d_{i},C+M+1\} & \textrm{if }j\geq|u_{1}|,\\
                                \delta(x)+d_{j}-d_{i} & \textrm{otherwise,}
                                \end{cases}\\
                            \\
                            \delta_t(x) =\begin{cases}
                                \min\{\delta(x)+d_{t}-d_{i},C+M+1\} & \textrm{if }t\geq|u_{1}|,\\
                                \delta(x)+d_{t}-d_{i} & \textrm{otherwise.}
                                \end{cases}
                            \end{array}
                            \]
                \end{itemize}
                
\medskip\noindent
    If $\varphi=\varphi_{1}\R \varphi_{2}$, then compute the value $n_\delta$ according to \eqref{eq-mdelta}, \eqref{eq-m2}, and \eqref{eq-n} and
     branch existentially to  the following two alternatives:

                    \begin{itemize}
                      \item branch universally to all triples from $\{(\rho(j),\delta_j,\varphi_{2})\mid i\leq j\leq |u_{1}|+n_\delta |u_{2}|\}$, where
                           \[
                            \begin{array}{l}
                            \delta_j(x)  =\begin{cases}
                                \min\{\delta(x)+d_{j}-d_{i},C+M+1\} & \textrm{if }j\geq|u_{1}|,\\
                                \delta(x)+d_{j}-d_{i} & \textrm{otherwise.}
                                \end{cases}
                            \end{array}
                            \]
                       \item branch existentially to each value $j \in [i+1,  |u_{1}|+n_\delta |u_{2}|]$, and then
                       branch universally to all triples from $\{(\rho(t),\delta_t,\varphi_{2})\mid i< t\leq j\}\cup \{(\rho(j),\delta_j,\varphi_{1}) \}$, where for all $x$:
                             \[
                            \begin{array}{l}
                            \delta_j(x) =\begin{cases}
                                \min\{\delta(x)+d_{j}-d_{i},C+M+1\} & \textrm{if }j\geq|u_{1}|,\\
                                \delta(x)+d_{j}-d_{i} & \textrm{otherwise,}
                                \end{cases}\\
                            \\
                            \delta_t(x) =\begin{cases}
                                \min\{\delta(x)+d_{t}-d_{i},C+M+1\} & \textrm{if }t\geq|u_{1}|,\\
                                \delta(x)+d_{t}-d_{i} & \textrm{otherwise.}
                                \end{cases}
                            \end{array}
                            \]
                      \end{itemize}
                    The machine $\tm$ clearly works in polynomial time.

                    Let us briefly discuss the necessary changes for the case $k=0$ (\ie, $w = u_{1}(u_{2})^{\omega}$).
         The main difficulty in the above algorithm is to find the upper bound of the witnessing position $j$ for the formulas
         $\varphi_{1}\U\varphi_{2}$ and $\varphi_{1}\R\varphi_{2}$. If $k=0$, then it is easily seen that for every $i\geq|u_{1}|$,
         formula $\varphi$ and valuation $\nu$, $(w,i,\nu)\models_\tptl\varphi$ if, and only if,  $(w,i+|u_{2}|,\nu)\models_\tptl\varphi$. One can easily see that the witnessing position
         $j$ can be bounded by $|u_{1}|+2|u_{2}|$.
         It is straightforward to implement the necessary changes in the above algorithm. \qed

         Note that one can easily adapt the proof of Theorem~\ref{Path check TPTL-upper-PSPACE} to obtain the following result for \emph{finite} data words:
           \begin{theorem}\label{PathcheckTPTL-upper-PSPACE_finite}
           	   Path checking for $\tptl_{\textup{b}}$ over finite binary encoded data words is in $\pspace$.
        \end{theorem}

    \paragraph{\bf Straight-line programs for data words.}
    
In this section we prove an extension of Theorem~\ref{Path check TPTL-upper-PSPACE}, where the data words
are succinctly specified by so called straight-line programs. Straight-line programs allow to represent data
words of exponential length. In Section~\ref{sec-one-counter} we will apply the result of this section to the model-checking problem
for deterministic one-counter machines with binary encoded updates.
    
 A {\em straight-line program}, briefly SLP, is a tuple $\mathcal{G}= (V,A_0,\rhs)$, where $V$ is a finite
 set of variables, $A_0 \in V$ is the output variable, and $\rhs$ (for right-hand side)
 is a mapping that associates with every variable $A \in V$ an expression $\rhs(A)$
 of the form $BC$, $B+d$, or $(P,d)$, where $B,C \in V$, $d \in \mathbb{N}$, and $P \subseteq \myprop$.
 Moreover, we require that the relation $\{ (A,B) \in V \times V \mid B \text{ occurs in } \rhs(A) \}$
 is acyclic. This allows to assign to every variable $A$ inductively a data word $\val_{\mathcal{G}}(A)$:
 \begin{itemize}
 \item  if $\rhs(A) \defeq (P,d)$ then $\val_{\mathcal{G}}(A) = (P,d)$, 
 \item if $\rhs(A) \defeq B+d$ then $\val_{\mathcal{G}}(A) = \val_{\mathcal{G}}(B)_{+d}$,  and
 \item  if $\rhs(A) \defeq BC$ then $\val_{\mathcal{G}}(A) = \val_{\mathcal{G}}(B)\val_{\mathcal{G}}(C)$. 
 \end{itemize}
 Finally, we set $\val(\mathcal{G}) = \val_{\mathcal{G}}(A_0)$.
 We assume that natural numbers appearing in 
 right-hand sides are binary encoded.  The size of the SLP $\mathcal{G}$ is defined as the sum of the 
 sizes of all right-hand sides, where a right-hand side of the form $BC$ (resp. $(P,d)$ or $B+d$) has size $1$
 (resp., $\lceil \log_2 d\rceil$).
 
 \begin{example}
Consider the straight-line program $\mathcal{G} = (\{ A_0, A_1, A_2, A_3, A_4,A_5\}, A_0, \rhs)$, where
$\rhs$ is defined as follows:
\begin{gather*}
\rhs(A_5) = (\{a,b\}, 2), \ \rhs(A_4) = (\{b,c\}, 3), \rhs(A_3) = A_5 A_4, \\
 \rhs(A_2) = A_3 A_3, \ \rhs(A_1) = A_2 + 4, \ \rhs(A_0) = A_2 A_1 .
\end{gather*}
We get
\begin{eqnarray*}
\val_{\mathcal{G}}(A_3) &=& (\{a,b\}, 2) (\{b,c\}, 3) \\
\val_{\mathcal{G}}(A_2) &=& (\{a,b\}, 2) (\{b,c\}, 3) (\{a,b\}, 2) (\{b,c\}, 3) \\
\val_{\mathcal{G}}(A_1) &=& (\{a,b\}, 6) (\{b,c\}, 7) (\{a,b\}, 6) (\{b,c\}, 7) \\
\val_{\mathcal{G}}(A_0) &=& (\{a,b\}, 2) (\{b,c\}, 3) (\{a,b\}, 2) (\{b,c\}, 3) (\{a,b\}, 6) (\{b,c\}, 7) (\{a,b\}, 6) (\{b,c\}, 7) \\
 &=&  \val(\mathcal{G}).
\end{eqnarray*}
\end{example}
 \begin{lemma}   \label{lemma-min-max}
 	 For a given SLP $\mathcal{G}$ we can compute the numbers $\min(\val(\mathcal{G}))$ and 
 $\max(\val(\mathcal{G}))$ in polynomial time.
 \end{lemma}   
 
 \begin{proof}
 Let $\mathcal{G}= (V,A_0,\rhs)$.
We compute $\min(\val_{\mathcal{G}}(A))$ for every $A \in V$ bottom-up according to the following rules:
\begin{itemize}
\item  If $\rhs(A) = (P,d)$, then $\min(\val_{\mathcal{G}}(A)) \defeq d$.
\item If $\rhs(A) = B+k$, then $\min(\val_{\mathcal{G}}(A)) \defeq \min(\val_{\mathcal{G}}(B)) + k$.
\item If $\rhs(A) = BC$, then $\min(\val_{\mathcal{G}}(A)) \defeq \min\{ \min(\val_{\mathcal{G}}(B)), \min(\val_{\mathcal{G}}(C)) \}$.
\end{itemize}
Of course, for computing $\max(\val(\mathcal{G}))$ in polynomial time one can proceed analogously.
\end{proof} 

 The following result extends Theorem~\ref{Path check TPTL-upper-PSPACE} to SLP-encoded infinite data words, \ie,
 infinite data words $u_{1}(u_{2})^{\omega}_{+k}$ that are succinctly represented by two SLPs for $u_1$ and $u_2$, respectively,  
 and the binary encoding of the number $k$.
 
  \begin{theorem}\label{Path check TPTL-upper-PSPACE-SLP}
  	  Path checking for $\tptl_{\textup{b}}$ over infinite (resp., finite) SLP-encoded data words  is in $\pspace$.
    \end{theorem}
    \begin{proof}
 
 %
%
 %
%
        We only show the statement for infinite data words.
        The proof is almost identical to the proof of Theorem~\ref{Path check TPTL-upper-PSPACE}.
        We explain the necessary adaptations. The input consists of a $\tptl$ formula $\psi$ (without loss
        of generality in negation normal form), two
        SLPs $\mathcal{G}_1$ and $\mathcal{G}_2$ and a binary encoded number $k$. Let $u_1 = \val(\mathcal{G}_1)$,
        $u_2 = \val(\mathcal{G}_2)$ and $w=u_{1}(u_{2})^{\omega}_{+k}$. We have to check whether 
        $w \models_\tptl \psi$ holds. For this we use the alternating polynomial time algorithm from 
        the proof of Theorem~\ref{Path check TPTL-upper-PSPACE}. We only consider the more difficult case
        $k > 0$. Define $C \defeq C_\psi$ by \eqref{C}, 
        $M \defeq M_{u_2}$ by \eqref{M}, $m_1$ by \eqref{eq-m1}, and $m_2$ by \eqref{eq-m2}.
        Note that the binary encodings of these four numbers can be computed in polynomial
        time by Lemma~\ref{lemma-min-max}. 
        
        As before, in every configuration, $\tm$ stores a triple
        $(i,\delta, \varphi)$, where $\varphi$ is a subformula of $\psi$, $i$ is a position in the data word from the interval $[0,|u_{1}|+|u_{2}|)$, and
        $\delta$ is a register valuation (with respect to the relative semantics) whose range is $[-m_1, \max\{m_1, C+M+1\}]$. For each such 
        $\delta$ we define $m_\delta$ as in \eqref{eq-mdelta} and $n_\delta$ as in \eqref{eq-n}. Given $\delta$, one can compute
        the binary encodings of these numbers in polynomial time. Furthermore, note that a triple $(i,\delta, \varphi)$ with the above 
        properties can be stored in polynomial space.

        The machine $\tm$ starts with the triple $(0,\mathbf{0},\psi)$, where $\mathbf{0}(x)=0$ for each register variable $x$.
        Then, $\tm$ branches according to the rules from the proof of Theorem~\ref{Path check TPTL-upper-PSPACE}. 
        The numbers $j$ and $t$ guessed there can be still stored in polynomial space, and all arithmetic manipulations
        can be done in polynomial time. In particular, 
        the function $\rho : \mathbb{N} \to [0, |u_{1}|+|u_{2}|)$ defined by
        $\rho(z) = z$ for $z  < |u_{1}|$, and $\rho(z) = ((z-|u_{1}|)\text{ mod }|u_{2}|)+|u_{1}|$ otherwise, can be computed in 
        polynomial time on binary encoded numbers.
\end{proof}
        
\subsection{Polynomial Time Upper Bounds for $\tptl$ with Fixed Number of Registers}
If the number of register variables is fixed and all numbers are unary encoded (or their unary encodings can be computed in polynomial time),
then the alternating Turing-machine in the proof of Theorem~\ref{Path check TPTL-upper-PSPACE}
works in logarithmic space. Since $\alogspace = \ptime$, we obtain the following statement for (i):

\begin{theorem}\label{Path check TPTL-upper-P}
     For every fixed $r \in \mathbb{N}$, path checking for $\tptl^r_{\textup{u}}$ over (i) infinite unary encoded data words or
     (ii)  infinite binary encoded monotonic data words is in $\ptime$.
    \end{theorem}
    \begin{proof}
    	    We start with the proof of the statement for (i).
In the algorithm from the proof of Theorem~\ref{Path check TPTL-upper-PSPACE},
if all numbers are given in unary, then the numbers $C+M+1$, $m_1$, $m_2$ and $n$ can be computed in
logarithmic space and are bounded polynomially in the input size. Moreover, a configuration of the form $(i, \delta, \varphi)$ needs only logarithmic space:
   clearly, the position  $i \in [0,|u_{1}|+|u_{2}|)$ and the subformula $\varphi$ only need logarithmic space. The valuation $\delta$ is an $r$-tuple
   over $[-m_1, \max\{m_1, C+M+1\}]$ and hence needs logarithmic space too, since $r$ is a constant. Hence, the alternating machine from
   the proof of Theorem~\ref{Path check TPTL-upper-PSPACE} works in logarithmic space. The theorem follows, since $\alogspace = \ptime$.
   
   Let us now prove statement (ii) in Theorem~\ref{Path check TPTL-upper-P}. We actually prove a slightly stronger statement for so called
   quasi-monotonic data words instead of monotonic data words.

 For a finite data word $u$, let $\min(u)$ (resp., $\max(u)$) be the minimal (resp., maximal) data value that occurs
 in $u$.  Given $k\in \N$, and two finite data words $u_1$ and $u_2$, we say that the infinite data word
 $u_1(u_2)^\omega_{+k}$ is {\em quasi-monotonic} if $\max(u_1) \leq \max(u_2) \leq  \min(u_2)+k$. Note that if $u_1(u_2)^\omega_{+k}$ is monotonic,
 then $u_1(u_2)^\omega_{+k}$ is also quasi-monotonic.
 
 Let us now prove statement (ii) (with ``monotonic'' replaced by ``quasi-monotonic'').
        Suppose that $k$, $u_1$ and $u_2$ are binary encoded, and $u_1(u_2)^\omega_{+k}$ is quasi-monotonic. The idea is that we construct in
        polynomial time two unary encoded finite data words $v_1,v_2$ and $l \in \N$ encoded in unary notation such that $u_1(u_2)^\omega_{+k}\models_\tptl \psi$ if, and only if, 
        $v_1(v_2)^\omega_{+l}\models_\tptl \psi$, where $\psi$ is the input $\tptl^r_{\textup{u}}$ formula. Then we can apply the alternating logarithmic space algorithm from the above
        proof for (i).

        Let $x_1\sim_1 c_1,\dots,x_m\sim_m c_m$ be all constraint formulas in $\psi$. Without loss of generality, we suppose that $c_i\leq c_{i+1}$ for $1\leq i<m$.
        We define an equivalence relation $\equiv_\psi$ on $\N$ such that $a\equiv_\psi b$ if $a=b$ or $a$ and $b$ both belong to one of the intervals
        $(-\infty,c_1)$, $(c_m,+\infty)$, $(c_i,c_{i+1})$ for some $1\leq i<m$.
        Define $C=\max\{|c_{1}|,\dots,|c_{m}|\}$.

        Suppose that $|u_1|=n_1$ and $|u_2|=n_2$.
        Let $d_1,\dots,d_{n_{1}+n_{2}}$ be an enumeration of all data values in $u_1u_2$  such that $d_j\leq d_{j+1}$ for $1\leq j < n_1+n_2$.
        Without loss of generality, we suppose that $d_1=0$.  For $1 < i \leq n_1+n_2$ let $\delta_i = d_i - d_{i-1}$.
        We define a new sequence $d'_1,\dots,d'_{n_{1}+n_{2}}$ inductively as follows: $d'_1 = 0$ and for all $1 < i \leq n_1+n_2$,
        $$
        d'_i =
        \begin{cases}
         d'_{i-1} +  \delta_i & \text{ if } \delta_i \leq C, \\
         d'_{i-1} +  C+1 & \text{ if } \delta_i > C .
         \end{cases}
        $$
        Intuitively, the data values $d'_j$ are obtained by shrinking the
        $d_j$ so that the largest difference between two different data values is bounded by $C+1$.

        We obtain the new data words $v_1$ and $v_2$ by replacing in $u_1$ and $u_2$ every data value $d_j$ by $d'_j$
        for every $j\in [1,n_1+n_2]$. Note that $d'_{n_1+n_2} \leq (C+1) \cdot (n_1+n_2-1)$.
        Since $C$ is given in unary notation, we can compute
        in polynomial time the unary encodings of the numbers $d_1',\dots,d_{n_{1}+n_{2}}'$.

       To define the number $l$, note that $\delta : =  \min(u_2)+k-\max(u_2)$ is the difference between the smallest
        data value in $(u_2)_{+k}$ and the largest data value in $u_2$ (which is the largest data value of $u_1u_2$).
        Since $u_1(u_2)^\omega_{+k}$ is quasi-monotonic,
        we have $\delta \geq 0$.
        We define the number $l$ as
        $$
        l =  \begin{cases}
          \max(v_2)-\min(v_2)+\delta & \text{ if } \delta \leq C, \\
         \max(v_2)-\min(v_2)+C+1 & \text{ if } \delta > C .
         \end{cases}
         $$
         Again,
        the unary encoding of $l$ can be computed in polynomial time.
        Let $e_i$ (resp., $e_i'$) be the data value in the $i$-th position of $u_1(u_2)^\omega_{+k}$ (resp., $v_1(v_2)^\omega_{+l}$).
        Then, for every $j>i$ we have $e_j-e_i\equiv_\psi e_j'-e_i'$.
        This implies that  $u_1(u_2)^\omega_{+k}\models_\tptl \varphi$ if, and only if,  $v_1(v_2)^\omega_{+l}\models_\tptl \varphi$. Finally, we can use the alternating logarithmic space algorithm to check whether $v_1(v_2)^\omega_{+l}\models_\tptl \varphi$ holds.
    \end{proof}
    For finite data words, we obtain a polynomial time algorithm also for binary encoded data words (assuming again a fixed
 number of register variables): 
 \begin{theorem} \label{thm-finite-path-binary}
 	 For every fixed $r \in \mathbb{N}$, path checking for $\tptl^r_{\textup{b}}$ over finite binary encoded data words is in  $\ptime$.
 \end{theorem}
 \begin{proof}
 	   Let the input data word $w$ be of length $n$ and let $d_1, \ldots, d_n$ be the data values appearing in $w$.
 	   Moreover, let $x_1, \ldots, x_r$ be the register variables appearing in the input formula $\psi$. Then, we only have
 	   to consider the $n^r$ many register valuation mappings $\nu : \{ x_1, \ldots, x_r \} \to \{d_1, \ldots, d_n \}$. For
 	   each of these mappings $\nu$, for every subformula $\varphi$ of $\psi$, and for every position $i$ in $w$ we
 	   check whether $(w,i,\delta) \models_\tptl \varphi$. This information is computed bottom-up (with respect to the
  structure of $\varphi$) in the usual way.
 \end{proof}
    For infinite data words we have to reduce the number of register variables to one in order to get
    a polynomial time complexity for  binary encoded numbers, \cf~Theorem~\ref{theorem-PSPACE-lower-bound-3}:

  \begin{theorem} \label{thm-1-reg-P}
  	  Path checking for $\tptl^1_{\textup{b}}$ over infinite binary encoded data words is in $\ptime$.
  \end{theorem}
    For the proof of Theorem~\ref{thm-1-reg-P} we first show some auxiliary results of independent interest.

\begin{lemma} \label{lemma-ltl-path}
For a given $\ltl$ formula $\psi$, words $u_1, \ldots, u_n, u \in (2^{\myprop})^*$ and binary encoded
numbers $N_1, \ldots, N_n \in \mathbb{N}$, the question whether $u_1^{N_1} u_2^{N_2} \dots u_n^{N_n}  u^\omega \models \psi$ holds,
belongs to $\ptime$ (more precisely,  $\textup{AC}^1(\textup{LogDCFL})))$.
\end{lemma}

\begin{proof}
	We first prove that 
	for all finite words $u,v \in  (2^{\myprop})^*$, every infinite word
$w \in (2^{\myprop})^\omega$ and every number $N \geq |\psi|$, we have $u v^N w \models \psi$
if, and only if, $u v^{|\psi|} w \models \psi$. 
For the proof, we use the Ehrenfeucht-Fra\"iss\'e 
game  for $\ltl$ (EF-game, for short), introduced in \cite{DBLP:journals/iandc/EtessamiW00} and briefly explained in the following.

Let $w_0=(P_0,d_0)(P_1,d_1)\dots$ and $w_1=(P'_0,d'_0)(P'_1,d'_1)\dots$ be two infinite data words over $2^\myprop$, and let $k\ge 0$.
The $k$-round EF-game is played by two players, called Spoiler and Duplicator, on $w_0$ and $w_1$. 
A game configuration is a pair of positions $(i_0,i_1)\in \N\times\N$, where $i_0$ is a position in $w_0$, and $i_1$ is a position in $w_1$. 
In each round of the game, if the current configuration is $(i_0,i_1)$, Spoiler  chooses an index $l\in \{0,1\}$ and a position $j_l>i_l$ in the data word $w_l$. Then, Duplicator responds with a position $j_{1-l}>i_{1-l}$ in the other data word $w_{1-l}$. 
After that, Spoiler has two options:
\begin{itemize}
\item He chooses to finish the current round. The current round then finishes with configuration  $(j_0,j_1)$. 
\item He chooses a position $j'_{1-l}$ so that $i_{1-l}<j'_{1-l}<j_{1-l}$. Then Duplicator responds with a position $j'_l$ so that $i_l <j'_l<j_l$. The current round then finishes with configuration $(j'_0,j'_1)$. 
\end{itemize}
The \emph{winning condition for Duplicator} is defined inductively by the number $k$ of rounds of the EF-game: 
duplicator wins the $0$-round EF-game starting in configuration $(i_0,i_1)$
 if $P_{i_0}=P'_{i_1}$. 
Duplicator wins the $k+1$-round EF-game  starting in configuration $(i_0,i_1)$
if $P_{i_0}=P'_{i_1}$  and for every choice of moves of Spoiler in the first round, Duplicator 
can respond such that Duplicator
wins the $k$-round EF-game starting in the finishing configuration of the first round (\ie, in the above configuration $(j_0,j_1)$ or $(j'_0, j'_1)$).

The \emph{until rank} of an $\ltl$ formula $\varphi$, denoted by $\urk(\varphi)$, is defined inductively on the structure of $\varphi$: 
\begin{itemize}
\item If $\varphi$ is $\true$, $p\in\myprop$ or $x\sim c$, then $\urk(\varphi)=0$. 
\item If $\varphi$ is $\neg\varphi_1$, then $\urk(\varphi)=\urk(\varphi_1)$. 
\item If $\varphi$ is $\varphi_1\wedge\varphi_2$, then $\urk(\varphi)=\max\{\urk(\varphi_1),\urk(\varphi_2)\}$. 
\item If $\varphi$ is $\varphi_1\U\varphi_2$, then $\urk(\varphi)=\max\{\urk(\varphi_1),\urk(\varphi_2)\}+1$. 
\end{itemize}
\begin{theorem}[\cite{DBLP:journals/iandc/EtessamiW00}]
	\label{theorem_etessami}
	Duplicator wins the $k$-round EF-game on $w_0$ and $w_1$ starting in configuration $(0,0)$ if, and only if, for every $\ltl$ formula $\varphi$ with $\urk(\varphi)\le k$, we have  $w_0\ltlmodels\varphi$ if, and only if, $w_1\ltlmodels\varphi$. 
\end{theorem}
Now, let $k\ge 0$, and assume that $w_0 =  u v^{m_0} w$ and $w_1 =  u v^{m_1} w$, where $m_0, m_1 \geq k$, $u$ and $v$ are finite words, and $w$ is an infinite word. 
One can easily prove that Duplicator can win the $k$-round EF-game starting from configuration $(0,0)$.
The point is that Duplicator can enforce that after the first round the new configuration $(i_0,i_1)$
satisfies one of the following two conditions:
\begin{itemize}
\item $w_0[{i_0\!:}] = w_1[{i_1\!:}]$, \ie, after the first round the suffix of $w_0$ starting in position $i_0$ is equal to the suffix of $w_1$ starting in $i_1$. This implies that Duplicator can win the remaining $(k-1)$-round EF-game starting in configuration $(i_0,i_1)$.  
\item $w_0[{i_0\!:}]$ (respectively, $w_1[{i_1\!:}]$) has the form $u' v^{n_0} w$ (respectively, $u' v^{n_1} w$), where $n_0, n_1 \geq k-1$.
	Hence, by induction, Duplicator can win the $(k-1)$-round EF-game from configuration $(i_0,i_1)$. 
\end{itemize}
By Theorem~\ref{theorem_etessami}, we have $u v^{m_0} w\ltlmodels\varphi$ if, and only if,  $u v^{m_1} w\ltlmodels\varphi$ for every $\ltl$ formula $\varphi$ with $\urk(\varphi)\le k$. 
It follows that for two infinite words $u_1^{m_1} u_2^{m_2} \dots u_n^{m_n}  u_{n+1}^\omega$ and $u_1^{m'_1} u_2^{m'_2} \dots u_n^{m'_n}  u_{n+1}^\omega$ satisfying $m_i,m'_i\ge k$ for all $i\in\{1,\dots, n\}$, we have $${u_1^{m_1} u_2^{m_2} \dots u_l^{m_n}  u_{n+1}^\omega\ltlmodels\varphi} \textrm{ if, and only if, }  {u_1^{m'_1} u_2^{m'_2} \dots u_l^{m'_n}  u_{n+1}^\omega\ltlmodels\varphi}$$ for every $\ltl$ formula $\varphi$ with $\urk(\varphi)\le k$.

We use this to prove the lemma as follows: replace every binary encoded exponent
$N_i$ in the word $u_1^{N_1} u_2^{N_2} \dots u_l^{N_n} u_{n+1}^\omega$ by $\min\{N_i, |\psi|\}$. 
By Theorem 3.6 in \cite{DBLP:conf/concur/MarkeyS03}, infinite path checking for $\ltl$ can be reduced in logspace to finite path checking for $\ltl$. Finite path checking for $\ltl$ is in $\textup{AC}^1(\textup{LogDCFL}) \subseteq \ptime$~\cite{DBLP:conf/icalp/KuhtzF09}. 
\end{proof}

\begin{lemma}\label{lemma-tptl-1}
Path checking for $\tptl_{\textup{b}}$ formulas, which do not contain the freeze quantifier $x.(\cdot)$, over
infinite binary encoded data words is in  $\ptime$ (in fact, $\textup{AC}^1(\textup{LogDCFL}))$.
\end{lemma}

\begin{proof}
	We reduce the question whether $w \models_\tptl \psi$ in logspace to an instance of the special $\ltl$ path-checking problem from
Lemma~\ref{lemma-ltl-path}.
Let $w  = u_{1}(u_{2})_{+k}^{\omega}$. 
We use $w[i]$ to denote the pair $(P_i,d_i)\in 2^{\myprop} \times\N$ occurring at the $i$-th position of $w$. 
Let $n_1 = |u_1|$ and $n_2 = |u_2|$. 
Without loss of generality we may assume that the only register variable that appears in constraint formulas in $\varphi$ is $x$. 
Note that since $\varphi$ does not contain the freeze quantifier, the value of $x$ is always assigned the initial value $d_0$.  

In order to construct an $\ltl$ formula from $\psi$, it remains to eliminate occurrences
of constraint formulas $x \sim c$ in $\psi$.  
Without loss of generality, we may assume that all constraints are of the form $x < c$ or $x > c$.
Let $x \sim_1 c_1, \ldots, x \sim_m c_m$ be a list of all
constraints that appear in $\psi$.
We introduce for every $j\in \{1, \dots,  m\}$ a new atomic
proposition $p_j$, and we define $\myprop' \defeq \myprop \cup \{ p_1, \ldots, p_m \}$.
Let $\psi'$ be obtained from $\psi$ by replacing every occurrence of
$x \sim_j c_j$ by $p_j$, and let  $w' \in (2^{\myprop'})^\omega$
be the $\omega$-word defined by $w'[i] = P_i \cup \{ p_j \mid 1 \leq j \leq m, d_i-d_0 \sim_j c_j \}$.
Clearly $w \models_\tptl \psi$ if, and only if, $w' \ltlmodels \psi'$.
Next we will show that the word $w'$ can be written in the form required in
Lemma~\ref{lemma-ltl-path}.

First of all, we can write $w'$ as
$w' = u'_1 u'_{2,0}  u'_{2,1} u'_{2,2} \dots$,
where $|u'_1| = n_1$ and $|u'_{2,i}| = n_2$.
The word $u'_1$ can be computed in logspace by evaluating all constraints
at all positions of $u_1$. Moreover, every word $u'_{2,i}$
is obtained from $u_2$ (without the data values) by adding the new propositions
$p_j$ at the appropriate positions.
Consider the equivalence relation $\equiv$ on $\mathbb{N}$ with
$a \equiv b$ if, and only if, $u'_{2,a} = u'_{2,b}$.
The crucial observations are that (i) every equivalence class of $\equiv$
is an interval, and (ii)
the index of $\equiv$ is bounded
by $1 + n_2 \cdot m$ (one plus the length $n_2$ of $u_2$ times the number $m$ of constraint formulas). To see this,
consider a position $i\in \{0,\dots,n_2-1\}$ in the word $u_2$ and a constraint $x \sim_j c_j$ for some $j\in\{1,\dots, m\}$. 
Then, the truth value of ``proposition $p_j$ is present at the $i$-th position of $u'_{2,x}$''
switches (from true to false or from false to true) at most once when $x$ grows. The reason for this is that the data value
at position $n_1 + i  + n_2 \cdot x$ is
$d_{n_1 + i  + n_2 \cdot x} = d_{n_1 + i} + k \cdot x$ for $x \geq 0$, i.e., it
grows monotonically with $x$. Hence, the truth value of $d_{n_1 + i} + k \cdot x - d_0 \sim_j c_j$
switches at most once, when $x$ grows.
So, we get at most $n_2 \cdot m$ many ``switching points'' in $\mathbb{N}$
which produce at most $1 + n_2 \cdot m$ many intervals.

Let $I_1, \ldots, I_l$ be a list of all $\equiv$-classes (intervals), where $a < b$ whenever $a \in I_i$, $b \in I_j$ and $i < j$.
The borders of these intervals can be computed in logspace
using arithmetics on binary encoded  numbers (addition, multiplication and division with remainder
can be carried out in logspace on binary encoded
numbers \cite{DBLP:journals/jcss/HesseAB02}). Hence, we can compute in logspace  the lengths $N_i = |I_i|$ of the intervals, where
$N_l = \omega$.
Also, for all $i\in\{1,\dots, l\}$  we can compute in logspace the unique word $v_i$ such that
$v_i = u'_{2,a}$ for all $a \in I_i$. Hence, $w' = u'_1 v_1^{N_1} \dots v_l^{N_l}$.
We can now apply Lemma~\ref{lemma-ltl-path}.
\end{proof}

\noindent
{\em Proof of Theorem~\ref{thm-1-reg-P}.}
Consider an infinite binary encoded data word $w=u_{1}(u_{2})_{+k}^{\omega}$ and a $\tptl^1_{\textup{b}}$ formula $\psi$.
Let $n = |u_1|+|u_2|$.
We check in polynomial time whether $w\models_\tptl\psi$.
A $\tptl$ formula $\varphi$ is \emph{closed} if
every occurrence of a register variable $x$ in $\varphi$ is under the scope of a freeze quantifier $x.(\cdot)$.
The proofs of the following two claims are straightforward:

\medskip
\noindent
\textit{Claim 1}: If $\varphi$ is closed, then for all valuations $\nu, \nu'$, $(w,i,\nu)\models_\tptl\varphi$ if, and only if,  $(w,i,\nu')\models_\tptl\varphi$.

\medskip
\noindent
\textit{Claim 2}: If $\varphi$ is closed and $i\geq|u_{1}|$, then for every valuation $\nu$, $(w,i,\nu)\models_\tptl\varphi$ if, and only if,  $(w,i+|u_{2}|,\nu)\models_\tptl\varphi$.

\medskip
\noindent
By Claim 1 we can write $(w,i)\models_\tptl\varphi$ for $(w,i,\nu)\models_\tptl\varphi$.
It suffices to compute for every (necessarily closed) subformula $x.\varphi$ of $\psi$ the set of all positions $i \in [0, n-1]$
such that $(w,i) \models_\tptl x.\varphi$, or equivalently $w[{i:}] \models_\tptl \varphi$.
We do this in a bottom-up process.
Consider a subformula $x.\varphi$ of $\psi$ and a position
$i \in [0, n-1]$. We have to check whether $w[{i:}] \models_\tptl \varphi$.
Let $x.\varphi_1, \ldots, x.\varphi_l$ be a list of all subformulas
of $\varphi$ that are not in the scope of another freeze quantifier within $\varphi$. We can assume that for
every $s\in\{1,\dots, l\}$ we have already determined the set of positions
$j\in  [0, n-1]$ such that $(w,j) \models_\tptl x.\varphi_s$. We can therefore replace every subformula $x. \varphi_s$
of $\varphi$ by a new atomic proposition $p_s$ and add in the data words $u_1$ (resp., $u_2$) the proposition $p_s$ to all
positions $j$ (resp., $j - |u_1|$) such that $(w,j) \models_\tptl x.\varphi_s$, where $j \in [0, n-1]$. Here, we make use
of Claim~2.  We denote the resulting formula and the resulting data word with $\varphi'$ and $w' = u'_{1}(u'_{2})_{+k}^{\omega}$,
respectively. Next, we explain how to compute from $u'_1$ and $u'_2$ new finite data words $v_1$ and $v_2$
such that $v_{1}(v_{2})_{+k}^{\omega} = w'[{i\!:}]$.  If $i < |u'_1|$  then we take $v_1 = u'_1[{i\!:}]$ and $v_2 = u'_2$.
If  $|u'_1| \leq i \leq n-1$, then we take $v_1 = \varepsilon$  and $v_2 = u'_2[{i\!:}] (u'_2[{:\!i-1}]+k)$, where $u'_2[{:\!i-1}]$ denotes the prefix of $u'_2$ up to position $i-1$. 
Finally, using Lemma~\ref{lemma-tptl-1} we can check in polynomial time whether
$w'[{i\!:}] \models_\tptl \varphi'$.
\qed

\section{Lower Complexity Bounds}
In this section we prove several $\ptime$-hardness and $\pspace$-hardness results for the path-checking problem. 
Together with the upper bounds from Section~\ref{sec-upper-bounds} we obtain sharp complexity results for the various
path-checking problems.

\subsection{$\ptime$-Hardness Results.}

  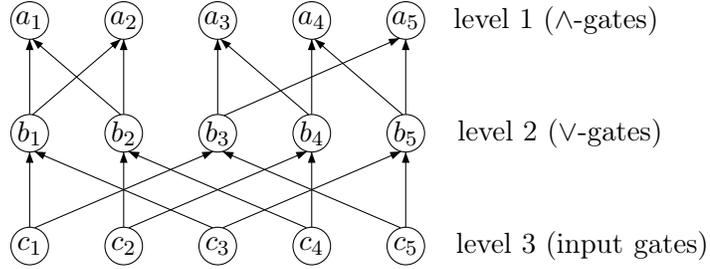
\begin{figure}[t]
    \setlength{\unitlength}{.5mm}
        \begin{center}
        \begin{picture}(160,65)(0,10)
        \gasset{Nw=10,Nh=10}
        \node(a1)(10,70){$a_1$}
        \node(a2)(35,70){$a_2$}
        \node(a3)(60,70){$a_3$}
        \node(a4)(85,70){$a_4$}
        \node(a5)(110,70){$a_5$}

        \node(b1)(10,40){$b_1$}
        \node(b2)(35,40){$b_2$}
        \node(b3)(60,40){$b_3$}
        \node(b4)(85,40){$b_4$}
        \node(b5)(110,40){$b_5$}

        \node(c1)(10,10){$c_1$}
        \node(c2)(35,10){$c_2$}
        \node(c3)(60,10){$c_3$}
        \node(c4)(85,10){$c_4$}
        \node(c5)(110,10){$c_5$}

        \drawedge(b1,a1){}
        \drawedge[sxo=-5,exo=5](b1,a2){}
        \drawedge[sxo=5,exo=-5](b2,a1){}
        \drawedge(b2,a2){}

        \drawedge(b3,a3){}
        \drawedge(b4,a4){}
        \drawedge(b5,a5){}
        \drawedge[sxo=5,exo=-5](b4,a3){}
        \drawedge[sxo=5,exo=-5](b5,a4){}
        \drawedge[sxo=-9,exo=9](b3,a5){}

        \drawedge(c2,b2){}
        \drawedge[sxo=-10,exo=10](c2,b4){}
        \drawedge(c4,b4){}
        \drawedge[sxo=10,exo=-10](c4,b2){}

        \drawedge(c1,b1){}
        \drawedge[sxo=10,exo=-10](c3,b1){}
        \drawedge[sxo=-10,exo=10](c3,b5){}
        \drawedge(c5,b5){}
        \drawedge[sxo=-10,exo=10](c1,b3){}
        \drawedge[sxo=10,exo=-10](c5,b3){}

        \node[Nframe=n](l1)(150,70){level 1 ($\wedge$-gates)}
        \node[Nframe=n](l2)(151,40){level 2 ($\vee$-gates)}
        \node[Nframe=n](l3)(157,10){level 3 (input gates)}
        \end{picture}
        \end{center}
        \caption{\label{fig-circuit} An SAM2-circuit with three levels.}
        \end{figure}

        We prove our $\ptime$-hardness results by a reduction from a restricted version of the Boolean circuit value problem. A \emph{Boolean circuit} is a finite 
        directed acyclic graph, where each node is called a \emph{gate}. 
        An \emph{input gate} is a node with indegree 0.
        All other gates are of a certain type, which is either $\vee$, $\wedge$ or $\neg$ (\ie, the corresponding logical OR, AND and NOT operations). 
        An \emph{output} gate is a node with outdegree 0.  A Boolean circuit is {\it monotone} if it does not have gates of type $\neg$.

        A \emph{synchronous alternating monotone circuit with fanin 2 and fanout 2} (SAM2-circuit, for short) is a monotone circuit divided into levels
        $1, \ldots, l$ for some $l \geq 2$ such that
        the following properties hold:
        \begin{itemize}
          \item All wires go from a gate in level $i+1$ to a gate from level $i$.
          \item All output gates are in  level $1$ and all input gates are in  level $l$.
          \item All gates in the same level are of the same type ($\wedge$, $\vee$ or input)
          	  and the types of the levels between $1$ and $l-1$ alternate between $\wedge$  and $\vee$. 
          \item All gates except the output gates have outdegree 2, and all gates except the input gates have indegree 2.
          	  For all $i\in\{2,\dots,l-1\}$, the two input gates for a gate at level $i$ are different.
        \end{itemize}

        By the restriction to fanin 2 and fanout 2, each level contains the same number of gates. Figure~\ref{fig-circuit} shows an example of an SAM2-circuit; the node names $a_i, b_i, c_i$ will be used later. 
        The \emph{circuit value problem} for SAM2-circuits, called SAM2CVP in \cite{Greenlaw:1995:LPC:203244}, is the following problem: given an SAM2-circuit $\alpha$, inputs $x_1,\dots,x_n\in \{0,1\}$, and a designated output gate $y$, does the output gate $y$ of $\alpha$ evaluate to $1$ on inputs $x_1,\dots,x_n$? 
        The circuit value problem for SAM2-circuits is $\ptime$-complete~ \cite{Greenlaw:1995:LPC:203244}.

        \begin{theorem} \label{thm-mtl-P}
        	Path checking for $\pure\mtl(\F,\X)_{\textup{u}}$ over finite unary encoded pure data words is $\ptime$-hard.
        \end{theorem}
        \begin{proof}
        	The proof is a reduction from SAM2CVP. 
        	Let $\alpha$ be an SAM2-circuit.  
        	We first encode each pair of consecutive levels of $\alpha$ into a data word, and combine these data words into a data word $w$,
       which is the encoding of the whole circuit. Then we construct a $\pure\mtl(\F,\X)_{\textup{u}}$ formula $\psi$ such that $w \models_\mtl \psi$ if, and only if,  $\alpha$ evaluates to $1$.
       The data word $w$ that we are going to construct contains gate names of $\alpha$ (and some copies of the gates) as atomic propositions.
       These propositions are only needed for the construction. At the end, we  remove all propositions from the data word $w$
       to obtain a pure data word.
       The whole construction can be done in logarithmic space. The reader might look at Example~\ref{ex-circuit}, where the construction is carried out for the circuit shown in Figure~\ref{fig-circuit}.

       Let $\alpha$ be an SAM2-circuit with $l \geq 2$ levels and $n$ gates in each level. By the restriction to fanin 2 and fanout 2, the induced undirected subgraph containing the nodes in level $i$ and level $i+1$ is comprised of several cycles; see Figure~\ref{inducedgra}.
       For instance, for the circuit in Figure~\ref{fig-circuit}, there are two cycles between level $1$ and $2$, and we have the same number of cycles between level $2$ and $3$.

        \begin{figure}[t]
        \setlength{\unitlength}{.7mm}
        \begin{center}
        \begin{picture}(170,40)
        \setlength{\unitlength}{0.7mm}
        \gasset{AHnb=0,Nw=2,Nh=2,ExNL=y}

        \node[NLdist=5](a1)(10,30){$a_{1,1}$}
        \node[NLdist=5](a2)(20,30){$a_{1,2}$}
        \node[NLdist=5](a3)(40,30){$a_{1,j_1}$}
        \imark[iangle=-60,ilength=8](a2)

        \node[NLdist=-5](b1)(10,10){$b_{1,1}$}
        \node[NLdist=-5](b2)(20,10){$b_{1,2}$}
        \node[NLdist=-5](b3)(40,10){$b_{1,j_1}$}
        \imark[iangle=120,ilength=8](b3)

        \drawedge(a1,b1){}
        \drawedge(a1,b2){}
        \drawedge(a2,b2){}
        \drawedge(a3,b3){}
        \drawedge(a3,b1){}
        \drawline[dash={0.5 1}0](26,20)(35,20)

        \node[NLdist=5](2a1)(55,30){$a_{2,1}$}
        \node[NLdist=5](2a2)(65,30){$a_{2,2}$}
        \node[NLdist=5](2a3)(85,30){$a_{2,j_2}$}
        \imark[iangle=-60,ilength=8](2a2)

        \node[NLdist=-5](2b1)(55,10){$b_{2,1}$}
        \node[NLdist=-5](2b2)(65,10){$b_{2,2}$}
        \node[NLdist=-5](2b3)(85,10){$b_{2,j_2}$}
        \imark[iangle=120,ilength=8](2b3)

        \drawedge(2a1,2b1){}
        \drawedge(2a1,2b2){}
        \drawedge(2a2,2b2){}
        \drawedge(2a3,2b3){}
        \drawedge(2a3,2b1){}
        \drawline[dash={0.5 1}0](71,20)(80,20)

        \drawline[dash={0.5 1}0](90,20)(100,20)

        \node[NLdist=5](3a1)(105,30){$a_{h,1}$}
        \node[NLdist=5](3a2)(115,30){$a_{h,2}$}
        \node[NLdist=5](3a3)(135,30){$a_{h,j_h}$}
        \node[Nframe=n](i)(149,30){level $i$}
        \imark[iangle=-60,ilength=8](3a2)

        \node[NLdist=-5](3b1)(105,10){$b_{h,1}$}
        \node[NLdist=-5](3b2)(115,10){$b_{h,2}$}
        \node[NLdist=-5](3b3)(135,10){$b_{h,j_h}$}
        \node[Nframe=n](i1)(153,9){level $i\!+\!1$}
        \imark[iangle=120,ilength=8](3b3)

        \drawedge(3a1,3b1){}
        \drawedge(3a1,3b2){}
        \drawedge(3a2,3b2){}
        \drawedge(3a3,3b3){}
        \drawedge(3a3,3b1){}
        \drawline[dash={0.5 1}0](121,20)(130,20)
        \end{picture}
        \end{center}
        \caption{The induced subgraph between level $i$ and $i+1$}
        \label{inducedgra}
        \end{figure}
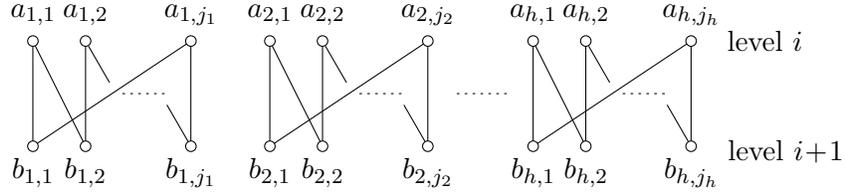

        We can enumerate in logarithmic space the gates of level $i$ and level $i+1$ such that they occur in the order shown in
        Figure~\ref{inducedgra}. To see this, let $a_1, \ldots, a_n$ (respectively, $b_1, \ldots, b_n$) be the nodes in level $i$ (respectively, $i+1$)
        in the order in which they occur in the input description. We start with $b_1$ and enumerate the nodes in the cycle containing
        $b_1$ (from $b_1$ we go to the smaller neighbour among $a_1, \ldots, a_n$. Then from this node the next node on the cycle is uniquely determined since the graph
        has degree 2. Thereby we store the current node in the cycle and the starting node $b_1$. As soon as we come back to $b_1$,
        we know that the first cycle is completed. To find the next cycle, we search for the first node from the list $b_2, \ldots, b_n$
        that is not reachable from $b_1$ (reachability in undirected graphs is in $\ls$
        	), and continue in this way.

        \newcommand{\cycles}{h}
        So, assume that the nodes in level $i$ and $i+1$ are ordered as in Figure~\ref{inducedgra}. Assume we have $\cycles$ cycles. Next, we want to get rid of the crossing edges between the rightmost node in level $i$ and the leftmost node in level $i+1$ in each cycle. For this, 
        we introduce for each cycle two dummy nodes which are basically copies of the leftmost node in each level. Formally,  for each         
        $t\in \{1,\dots,\cycles\}$, add a node $a_{t,1}'$ (respectively, $b_{t,1}'$) after $a_{t,j_t}$ (respectively, $b_{t,j_t}$), and then replace the edge $(a_{t,j_t}, b_{t,1})$ by a new edge $(a_{t,j_t}, b_{t,1}')$.
        In this way we  obtain the graph
        shown in Figure~\ref{PathPic}. Again, the construction can be done in logarithmic space.

        \begin{figure}[t]
        \setlength{\unitlength}{.7mm}
        \begin{center}
        \begin{picture}(170,40)
        \gasset{AHnb=0,Nw=2,Nh=2,ExtNL=y}

        \node[NLdist=2](a1)(10,30){$a_{1,1}$}
        \node[NLdist=2](a2)(20,30){$a_{1,2}$}
        \node[NLdist=2](a3)(40,30){$a_{1,j_1}$}
        \node[Nfill=y,NLdist=2](a4)(50,30){$a_{1,1}'$}
        \imark[iangle=-60,ilength=8](a2)

        \node[NLdist=-9](b1)(10,10){$b_{1,1}$}
        \node[NLdist=-9](b2)(20,10){$b_{1,2}$}
        \node[NLdist=-9](b3)(40,10){$b_{1,j_1}$}
        \node[Nfill=y,NLdist=-9](b4)(50,10){$b_{1,1}'$}
        \imark[iangle=120,ilength=8](b3)

        \drawedge(a1,b1){}
        \drawedge(a1,b2){}
        \drawedge(a2,b2){}
        \drawedge(a3,b3){}
        \drawline[dash={0.5 1}0](26,20)(35,20)

        \node[NLdist=2](2a1)(60,30){$a_{2,1}$}
        \node[NLdist=2](2a2)(70,30){$a_{2,2}$}
        \node[NLdist=2](2a3)(90,30){$a_{2,j_2}$}
        \node[Nfill=y,NLdist=2](2a4)(100,30){$a_{2,1}'$}
        \imark[iangle=-60,ilength=8](2a2)

        \node[NLdist=-9](2b1)(60,10){$b_{2,1}$}
        \node[NLdist=-9](2b2)(70,10){$b_{2,2}$}
        \node[NLdist=-9](2b3)(90,10){$b_{2,j_2}$}
        \node[Nfill=y,NLdist=-9](2b4)(100,10){$b_{2,1}'$}
        \imark[iangle=120,ilength=8](2b3)

        \drawedge(2a1,2b1){}
        \drawedge(2a1,2b2){}
        \drawedge(2a2,2b2){}
        \drawedge(2a3,2b3){}
        \drawline[dash={0.5 1}0](76,20)(85,20)

        \drawedge(b4,a3){}
        \drawedge(2a3,2b4){}

        \drawline[dash={0.5 1}0](105,20)(114,20)

        \node[NLdist=2](3a1)(120,30){$a_{h,1}$}
        \node[NLdist=2](3a2)(130,30){$a_{h,2}$}
        \node[NLdist=2](3a3)(150,30){$a_{h,j_h}$}
        \node[Nfill=y,NLdist=2](3a4)(160,30){$a_{h,1}'$}
        \imark[iangle=-60,ilength=8](3a2)

        \node[NLdist=-9](3b1)(120,10){$b_{h,1}$}
        \node[NLdist=-9](3b2)(130,10){$b_{h,2}$}
        \node[NLdist=-9](3b3)(150,10){$b_{h,j_h}$}
        \node[Nfill=y,NLdist=-9](3b4)(160,10){$b_{h,1}'$}
        \imark[iangle=120,ilength=8](3b3)

        \drawedge(3a1,3b1){}
        \drawedge(3a1,3b2){}
        \drawedge(3a2,3b2){}
        \drawedge(3a3,3b3){}
        \drawedge(3a3,3b4){}
        \drawline[dash={0.5 1}0](136,20)(145,20)
        \end{picture}
        \end{center}
        \caption{The graph after adding dummy nodes}
        \label{PathPic}
        \end{figure}
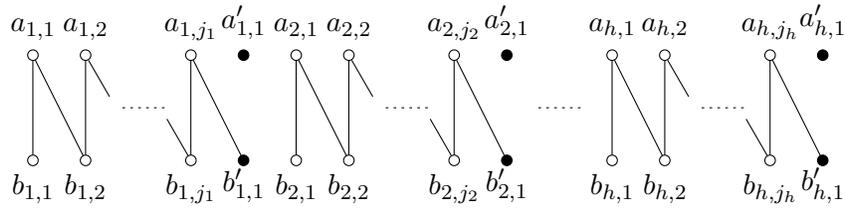

         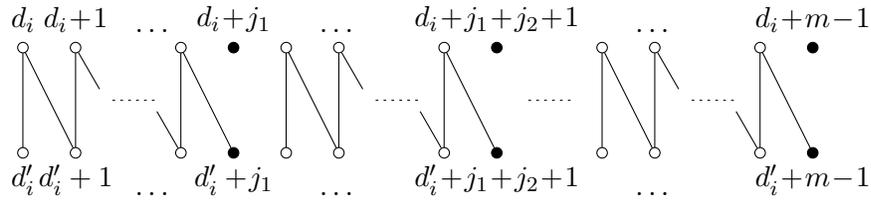
\begin{figure}[t]
        \setlength{\unitlength}{.7mm}
        \begin{center}
        \begin{picture}(170,40)
        \gasset{AHnb=0,Nw=2,Nh=2,ExtNL=y}

        \node[NLdist=2](a1)(10,30){$d_i$}
        \node[NLdist=2](a2)(20,30){$d_i\!+\!1$}
        \node[NLdist=2](a3)(40,30){$\dots\qquad$}
        \node[Nfill=y,NLdist=2](a4)(50,30){$d_i\!+\!j_1$}
        \imark[iangle=-60,ilength=8](a2)

        \node[NLdist=-9](b1)(10,10){$d'_i$}
        \node[NLdist=-9](b2)(20,10){$d'_i+1$}
        \node[NLdist=-9](b3)(40,10){$\dots\qquad$}
        \node[Nfill=y,NLdist=-9](b4)(50,10){$d'_i+\!j_1$}
        \imark[iangle=120,ilength=8](b3)

        \drawedge(a1,b1){}
        \drawedge(a1,b2){}
        \drawedge(a2,b2){}
        \drawedge(a3,b3){}
        \drawline[dash={0.5 1}0](26,20)(35,20)

        \node(2a1)(60,30){}
        \node[NLdist=2](2a2)(70,30){$\dots$}
        \node(2a3)(90,30){}
        \node[Nfill=y,NLdist=2](2a4)(100,30){$d_i\!+\!j_1\!+\!j_2\!+\!1$}
        \imark[iangle=-60,ilength=8](2a2)

        \node(2b1)(60,10){}
        \node[NLdist=-9](2b2)(70,10){$\dots$}
        \node(2b3)(90,10){}
        \node[Nfill=y,NLdist=-9](2b4)(100,10){$d'_i\!+\!j_1\!+\!j_2\!+\!1$}
        \imark[iangle=120,ilength=8](2b3)

        \drawedge(2a1,2b1){}
        \drawedge(2a1,2b2){}
        \drawedge(2a2,2b2){}
        \drawedge(2a3,2b3){}
        \drawline[dash={0.5 1}0](76,20)(85,20)

        \drawedge(b4,a3){}
        \drawedge(2a3,2b4){}

        \drawline[dash={0.5 1}0](105,20)(114,20)

        \node(3a1)(120,30){}
        \node[NLdist=2](3a2)(130,30){$\dots$}
        \node(3a3)(150,30){}
        \node[Nfill=y,NLdist=2](3a4)(160,30){$d_i\!+\!m\!-\!1$}
        \imark[iangle=-60,ilength=8](3a2)

        \node(3b1)(120,10){}
        \node[NLdist=-9](3b2)(130,10){$\dots$}
        \node(3b3)(150,10){}
        \node[Nfill=y,NLdist=-9](3b4)(160,10){$d'_i\!+\!m\!-\!1$}
        \imark[iangle=120,ilength=8](3b3)

        \drawedge(3a1,3b1){}
        \drawedge(3a1,3b2){}
        \drawedge(3a2,3b2){}
        \drawedge(3a3,3b3){}
        \drawedge(3a3,3b4){}

        \drawline[dash={0.5 1}0](136,20)(145,20)
        \end{picture}
        \end{center}
        \caption{The graph after assigning data values to the nodes}
        \label{GateNumber}
        \end{figure}

        By adding dummy nodes, we can assume that for every $i\in\{1,\dots,l-1\}$, the subgraph between level $i$ and $i+1$
        has the same number (say $h$) of cycles (this is only done for notational convenience, and we still suppose that there are $n$ gates in each level).
        Thus, after the above step we have $m = n+h$ many nodes in each level.
        Define $d_i = (i-1)\cdot 2m$ and $d'_i = d_i+m$.
        Next we are going to label the nodes from  Figure~\ref{PathPic}:  in level $i$ (respectively, $i+1$) with the numbers
        $d_i,d_i+1,\dots,d_i+m-1$  (respectively $d'_i,d'_i+1\dots, d'_i+m-1$)
        in this order,  see Figure~\ref{GateNumber}.
        Note that this labelling is the crucial point for encoding the wiring between the gates of two levels: the difference between two connected nodes in level $i$ and level $i+1$ is always $m$ or $m+1$.
        We will later exploit this fact and use the modality $\F_{[m,m+1]}$ (respectively, $\G_{[m,m+1]}$) to jump from a $\vee$-gate (respectively, $\wedge$-gate) in
        level $i$ to a successor gate in level $i+1$.
        We now obtain in logarithmic space the data word $w_i = w_{i,1} w_{i,2}$, where
         \begin{eqnarray*}
         w_{i,1} &=& \begin{cases}
         \displaystyle (a_{1,1},d_i)(a_{1,2},d_i+1)\cdots(a_{1,j_1},d_i+j_1-1) & \\
          \displaystyle (a_{2,1},d_i+j_1+1)(a_{2,2},d_i+j_1+2)\cdots(a_{2,j_2},d_i+j_1+j_2) \cdots & \\
          \displaystyle (a_{h,1},d_i+\sum_{t=1}^{h-1} j_t +h-1)(a_{h,2},d_i+\sum_{t=1}^{h-1} j_t +h)\cdots(a_{h,j_h},d_i+m-2) &
          \end{cases}
          \end{eqnarray*}
          \begin{eqnarray*}
         w_{i,2} &=& \begin{cases}
         \displaystyle (b_{1,1},d'_i)\cdots(b_{1,j_1},d'_i+j_1-1)(b_{1,1}',d'_i+j_1) & \\
          \displaystyle (b_{2,1},d'_i+j_1+1)\cdots(b_{2,j_2},d'_i+j_1+j_2)(b_{2,1}',d'_i+j_1+j_2+1) \cdots & \\
          \displaystyle (b_{h,1},d'_i+\sum_{t=1}^{h-1} j_t +h-1)\cdots(b_{h,j_h},d'_i+m-2)(b_{h,1}',d'_i+m-1)&
          \end{cases}
        \end{eqnarray*}
        which is the encoding of the wires between level $i$ and level $i+1$ from Figure~\ref{GateNumber}.
        Note that the new nodes $a_{1,1}',a_{2,1}',\dots,a_{h,1}'$ in level $i$
        of the graph in Figure~\ref{PathPic} do not occur in $w_{i,1}$.

        Suppose now that for all $i\in\{1,\dots,l-1\}$, the data words $w_i$ is constructed. 
        We combine $w_1, w_2, \dots, w_{l-1}$ to obtain the data word $w$ that encodes the complete circuit
        as follows. Suppose that
        \[
        w_{i,2} = (\tilde  b_1,y_1)\dots(\tilde b_{m},y_{m}) \text{ and } w_{i+1,1} = (b_1,z_1)\dots(b_n,z_n) .
        \]
        Note that every $\tilde b_s$ is either one of the $b_j$ or $b'_j$ (the copy of $b_j$).
        Let
        \[
        v_{i+1,1} = (\tilde  b_1,z'_1) \dots (\tilde b_{m},z'_{m}),
        \]
        where the data values $z'_s$ are determined as follows:
        if $\tilde b_s = b_j$ or $\tilde b_s = b'_j$, then $z'_s  = z_j$.
        Then, the data word $w$ is
        $w = w_{1,1} w_{1,2} v_{2,1} w_{2,2}  \dots v_{l-1,1} w_{l-1,2}$.

        Next, we explain the idea of how to construct the $\mtl$ formula. 
        Consider a gate $a_j$ of level $i$ for some $i\in\{2,\dots,l-1\}$, and assume that level $i$ consists of $\vee$-gates.
        Let $b_{j_1}$ and $b_{j_2}$ (from level $i+1$) be the two input gates for $a_j$.
        In the above data word $v_{i,1}$ there is a unique position
        where the proposition $a_j$ occurs, and possibly a position where the copy $a'_j$ occurs.
        If both positions exist, then they carry the same data value.
         Let us point to one of these positions. Using an $\mtl$ formula, we
        want to branch (existentially) to the positions in the factor $v_{i+1,1}$, where the propositions
        $b_{j_1}, b'_{j_1}, b_{j_2}, b'_{j_2}$ occur (where $b'_{j_1}$ and $b'_{j_2}$ possibly do not exist).
        For this, we use the modality $\F_{[m,m+1]}$. By construction,
        this modality branches existentially to  positions in the factor $w_{i,2}$, where
        the propositions   $b_{j_1}, b'_{j_1}, b_{j_2}, b'_{j_2}$ occur. Then, using the iterated next modality
         $\X^{m}$, we jump to the corresponding positions in $v_{i+1,1}$.

         In the above argument, we assumed that $i\in\{2,\dots, l-1\}$. If $i=1$, then we can argue similarly, if we
         assume that we are pointing to the unique $a_j$-labelled position of the prefix $w_{1,1}$ of $w$.
         Now consider level $l-1$.
        Suppose that
        $$w_{l-1,2} = (\tilde e_1,d_1)\dots (\tilde e_{m},d_{m}).$$
        Let $e_1, \ldots, e_n$ be the original gates of level $l$, which
        all belong to  $\{\tilde e_1, \ldots, \tilde e_{m}\}$, and let $x_i \in \{0,1\}$ be the input value for
        gate $e_i$.
        Define
        \begin{equation} \label{eq-set-I}
        	J=\{ j \mid j \in [1,m], \exists i \in [1,n] : \tilde e_j \in \{ e_i, e'_i \}, x_i = 1 \}.
        \end{equation}
Let the designated output gate be the $k$-th node in level $1$.
        We construct the $\pure\mtl(\F,\X)$ formula $\psi=\X^{k-1} \varphi_{1}$, where, for every $i\in\{1,\dots,l-1\}$, the formula $\varphi_i$  is  inductively defined as follows:
        \[
        	\varphi_{i}\defeq\begin{cases}
        \F_{[m,m+1]} \X^{m} \varphi_{i+1} &\text{if $i< l-1$ and level $i$ is a $\vee$-level,}\\[1mm]
        \G_{[m,m+1]} \X^{m} \varphi_{i+1} &\text{if $i< l-1$ and level $i$ is a $\wedge$-level,}\\[1mm]
        \F_{[m,m+1]}(\bigvee_{j \in J} \X^{m-j} \neg \X \,\true) & \textrm{if } i= l-1 \text{ and level $i$ is a $\vee$-level,}\\[1mm]
        \G_{[m,m+1]}(\bigvee_{j \in J} \X^{m-j} \neg \X \,\true) & \textrm{if } i= l-1 \text{ and level $i$ is a $\wedge$-level.}
        \end{cases}
        \]
        The formula $\neg\X \,\true$ is only true in the last position of a data word. Suppose data word $w$ is the encoding of the circuit.
        From the above consideration, it follows that  $w \models_\mtl \psi$ if, and only if,  the circuit $\alpha$ evaluates to $1$.
        Note that we only use the unary modalities $\F,\G,\X$ and do not use any propositions in $\psi$. We can thus ignore the propositional part in the data word $w$ to get a pure data word. Since the number $m$ is bounded by $2n$, and all data values in $w$ are bounded by $4nl$, where $n$ is the number of gates in each level and $l$ is the number of levels, we can compute the formula $\psi$ and data word $w$ where the interval borders and data values are encoded in unary notation in logarithmic space.
        \end{proof}

        \begin{example} \label{ex-circuit}
        Let $\alpha$ be the SAM2-circuit from Figure~\ref{fig-circuit}. It has 3 levels and 5 gates in each level. Level 1 contains $\wedge$-gates and level 2 contains $\vee$-gates. There are 2 cycles in the subgraph between level 1 and  2, and also 2 cycles in the subgraph between level 2 and 3.
        The encoding for level 1 and level 2 is
        \begin{equation}\label{Enc12}
        \begin{split}
         (a_1,0)(a_2,1)(a_3,3)(a_4,4)(a_5,5)\\
         (b_1,7)(b_2,8)(b_1',9)(b_3,10)(b_4,11)(b_5,12)(b_3',13) &,
        \end{split}
        \end{equation}
        which can be obtained from Figure~\ref{Num12}. The new nodes $a_1'$ and $a_3'$ in level 1 are not used for the final encoding in the data word.
        \begin{figure}[t]
         \setlength{\unitlength}{.5mm}
         \begin{center}
        \begin{picture}(150,55)(10,10)
        \gasset{Nw=10,Nh=10}
        \node(a1)(10,50){$a_1$}
        \node(a2)(35,50){$a_2$}
        \node[Nfill=y,fillgray=0.7](a)(60,50){$a_1'$}
        \node(a3)(85,50){$a_3$}
        \node(a4)(110,50){$a_4$}
        \node(a5)(135,50){$a_5$}
        \node[Nfill=y,fillgray=0.7](a6)(160,50){$a_3'$}

        \node[Nframe=n](n1)(10,60){0}
        \node[Nframe=n](n2)(35,60){1}
        \node[Nframe=n](n3)(60,60){2}
        \node[Nframe=n](n4)(85,60){3}
        \node[Nframe=n](n5)(110,60){4}
        \node[Nframe=n](n6)(135,60){5}
        \node[Nframe=n](n7)(160,60){6}

        \node[Nframe=n](n7)(10,10){7}
        \node[Nframe=n](n8)(35,10){8}
        \node[Nframe=n](n9)(60,10){9}
        \node[Nframe=n](n10)(85,10){10}
        \node[Nframe=n](n11)(110,10){11}
        \node[Nframe=n](n12)(135,10){12}
        \node[Nframe=n](n13)(160,10){13}

        \node(b1)(10,20){$b_1$}
        \node(b2)(35,20){$b_2$}
        \node[Nfill=y,fillgray=0.7](b11)(60,20){$b_1'$}
        \node(b3)(85,20){$b_3$}
        \node(b4)(110,20){$b_4$}
        \node(b5)(135,20){$b_5$}
        \node[Nfill=y,fillgray=0.7](b33)(160,20){$b_3'$}

        \drawedge[AHnb=0](b1,a1){}
        \drawedge[AHnb=0](b2,a2){}
        \drawedge[AHnb=0](b3,a3){}
        \drawedge[AHnb=0](b4,a4){}
        \drawedge[AHnb=0](b5,a5){}

        \drawedge[AHnb=0,sxo=6,exo=-6](b2,a1){}
        \drawedge[AHnb=0,sxo=6,exo=-6](b11,a2){}
        \drawedge[AHnb=0,sxo=6,exo=-6](b4,a3){}
        \drawedge[AHnb=0,sxo=6,exo=-6](b5,a4){}
        \drawedge[AHnb=0,sxo=6,exo=-6](b33,a5){}
        \end{picture}
        \end{center}
        \caption{The labelling for level 1 and  2}
        \label{Num12}
        \end{figure}
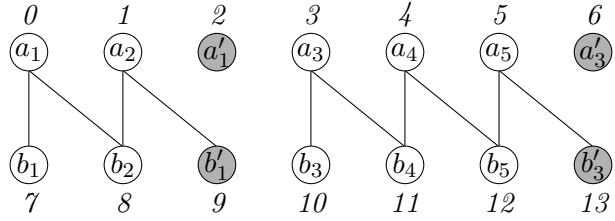
        The encoding for level 2 and 3 is
        \begin{equation}\label{Enc23}
        \begin{split}
        (b_1,14)(b_5,15)(b_3,16)(b_2,18)(b_4,19)\\
        (c_1,21)(c_3,22)(c_5,23)(c_1',24)(c_2,25)(c_4,26)(c_2',27)& ,
        \end{split}
        \end{equation}
        which can be obtained from Figure~\ref{Num23}. We skip the new nodes $b_1'$ and $b_2'$ in level 2 in this encoding.

        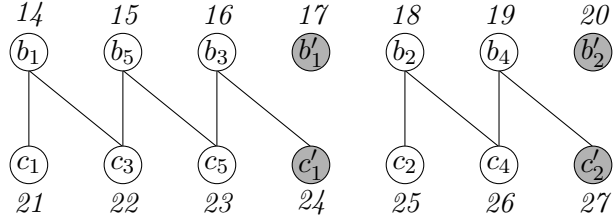
\begin{figure}[t]
         \setlength{\unitlength}{.5mm}
         \begin{center}
        \begin{picture}(150,55)(10,10)
        \gasset{Nw=10,Nh=10}
        \node(b1)(10,50){$b_1$}
        \node(b5)(35,50){$b_5$}
        \node(b3)(60,50){$b_3$}
        \node[Nfill=y,fillgray=0.7](b11)(85,50){$b_1'$}
        \node(b2)(110,50){$b_2$}
        \node(b4)(135,50){$b_4$}
        \node[Nfill=y,fillgray=0.7](b22)(160,50){$b_2'$}

        \node[Nframe=n](n1)(10,60){14}
        \node[Nframe=n](n2)(35,60){15}
        \node[Nframe=n](n3)(60,60){16}
        \node[Nframe=n](n4)(85,60){17}
        \node[Nframe=n](n5)(110,60){18}
        \node[Nframe=n](n6)(135,60){19}
        \node[Nframe=n](n7)(160,60){20}

        \node[Nframe=n](n7)(10,10){21}
        \node[Nframe=n](n8)(35,10){22}
        \node[Nframe=n](n9)(60,10){23}
        \node[Nframe=n](n10)(85,10){24}
        \node[Nframe=n](n11)(110,10){25}
        \node[Nframe=n](n12)(135,10){26}
        \node[Nframe=n](n13)(160,10){27}

        \node(c1)(10,20){$c_1$}
        \node(c3)(35,20){$c_3$}
        \node(c5)(60,20){$c_5$}
        \node[Nfill=y,fillgray=0.7](c55)(85,20){$c_1'$}
        \node(c2)(110,20){$c_2$}
        \node(c4)(135,20){$c_4$}
        \node[Nfill=y,fillgray=0.7](c44)(160,20){$c_2'$}

        \drawedge[AHnb=0](b1,c1){}
        \drawedge[AHnb=0](b2,c2){}
        \drawedge[AHnb=0](b5,c3){}
        \drawedge[AHnb=0](b4,c4){}
        \drawedge[AHnb=0](c5,b3){}

        \drawedge[AHnb=0,sxo=6,exo=-6](c5,b5){}
        \drawedge[AHnb=0,sxo=6,exo=-6](c3,b1){}
        \drawedge[AHnb=0,sxo=6,exo=-6](c55,b3){}
        \drawedge[AHnb=0,sxo=6,exo=-6](c4,b2){}
        \drawedge[AHnb=0,sxo=6,exo=-6](c44,b4){}
        \end{picture}
        \end{center}
        \caption{The labelling for level 2 and 3}
        \label{Num23}
        \end{figure}
        We combine (\ref{Enc12}) and (\ref{Enc23}) to obtain the following data word (\ref{Enc123}) which is the encoding of the circuit $\alpha$. The encoding for level 1 and 2 determines the order of the propositional part of the third line in (\ref{Enc123}), and the encoding for level 2 and 3 determines its data values.
        \begin{equation}\label{Enc123}
            \begin{split}
        (a_1,0)(a_2,1)(a_3,3)(a_4,4)(a_5,5)\\
        (b_1,7)(b_2,8),(b_1',9)(b_3,10)(b_4,11)(b_5,12)(b_3',13)\\
        (b_1,14)(b_2,18)(b_1',14)(b_3,16)(b_4,19)(b_5,15)(b_3',16)\\
        (c_1,21)(c_3,22)(c_5,23)(c_1',24)(c_2,25)(c_4,26)(c_2',27)& .
          \end{split}
        \end{equation}
        Let the designated output gate be $a_3$ in level 1, and assume that the input
        gates $c_1, c_4, c_5$ (respectively, $c_2, c_3$)
        receive the value $0$ (respectively, $1$). Then the set $J$ from \eqref{eq-set-I} is
        $J=\{2,5,7\}$ and the formula $\psi$ is
         \[\psi=\X^2\big(\G_{[7,8]}\X^7\big(F_{[7,8]}\big(\bigvee_{j\in \{2,5,7\}}\X^{7-j} \neg \X\, \true\big)\big)\big).\]
        \end{example}
         \begin{theorem} \label{thm-mtl-P-infinite}
         	 Path checking for $\mtl(\F,\X)_{\textup{u}}$ over infinite unary encoded data words is $\ptime$-hard.
        \end{theorem}

        \begin{proof}
        	We use an adaptation of the proof for Theorem~\ref{thm-mtl-P}. Let $p$ be an atomic proposition that is not used in the data word $w$ defined in the proof of Theorem~\ref{thm-mtl-P}. Define the infinite data word $w'=w\,(p,5ml)^\omega_{+0}$, and redefine for every $i\in\{1,\dots,l-1\}$ the formula $\varphi_i$ by:
        \[
        	\varphi_{i}\defeq \begin{cases}
        \F_{[m,m+1]} \X^{m} \varphi_{i+1} &\text{if $i< l-1$ and level $i$ is a $\vee$-level,}\\[1mm]
        \G_{[m,m+1]} \X^{m} \varphi_{i+1} &\text{if $i< l-1$ and level $i$ is a $\wedge$-level,}\\[1mm]
        \F_{[m,m+1]}(\bigvee_{j \in I} \X^{m-j} (\neg p\wedge \X p)) & \textrm{if } i= l-1 \text{ and level $i$ is a $\vee$-level,}\\[1mm]
        \G_{[m,m+1]}(\bigvee_{j \in I} \X^{m-j} (\neg p\wedge \X p)) & \textrm{if } i= l-1 \text{ and level $i$ is a $\wedge$-level.}
        \end{cases}
        \]
        Then $w' \models_\mtl \psi$ if, and only if, the circuit $\alpha$ evaluates to $1$.
        \end{proof}
        Note that the construction in the proof of Theorem~\ref{thm-mtl-P} uses  data words that are not monotonic. This is indeed unavoidable:  
        path checking for $\mtl$ over finite monotonic data words is in  $\textup{AC}^1(\textup{LogDCFL})$~\cite{DBLP:conf/icalp/BundalaO14}. 
        In the following, we prove that in contrast to $\mtl$, the path-checking problem for $\tptl^1_{\textup{u}}$ over finite monotonic data words is $\ptime$-hard, even for the pure fragment that only uses only the unary modalities $\F$ and $\X$. 
         \begin{theorem}\label{corPhard}
         	 Path checking for $\pure\tptl^1(\F,\X)_{\textup{u}}$ over finite unary encoded strictly monotonic pure data words is $\ptime$-hard.
        \end{theorem}
        Before we prove Theorem~\ref{corPhard}, we prove  $\ptime$-hardness of path checking for some extension of $\mtl$.

        In $\mtl$, the $\U$ modality is annotated by some interval $I$. 
        If, instead, we allow the $\U$ modality to be annotated by a finite union of intervals $I_{1}\cup I_2 \cup \dots\cup I_{n}$, then we call this logic \emph{succinct} $\mtl$, $\rmtl$ for short. 
        Formally, the syntax and semantics of $\rmtl$ is the same as for $\mtl$, except that the set $I$ in $\U_I$ can
be a finite union $I =  I_{1}\cup I_2 \cup \dots\cup I_{n}$ of intervals $I_i \subseteq \mathbb{Z}$.
The corresponding fragments of $\rmtl$ are defined in the expected way. 

        Let $I =  \bigcup_{i=1}^n  I_i$. It is easily seen that ($\equiv$ denotes logical equivalence)
        \begin{gather*}
                \varphi_1\U_I\varphi_2 \equiv   \bigvee_{i=1}^n \varphi_1 \U_{I_i}\varphi_2,  \text{ and }
        \varphi_1\U_{I} \varphi_2 \equiv  x.\varphi_1 \U \big(\big(\bigvee_{i=1}^n x\in I_i\big) \wedge \varphi_2\big).
        \end{gather*}
The following two propositions are easy to prove.

\begin{proposition}
        Each $\rmtl$ formula is equivalent to an $\mtl$ formula which can be exponentially larger.
        \end{proposition}

        \begin{proposition}\label{fact2}
        	For a given  $\rmtl$ formula $\varphi$  one can compute in logspace an equivalent $\tptl^1$ formula of size polynomial in the size of $\varphi$.
        \end{proposition}

        In the following, we prove $\ptime$-hardness of path checking for $\rmtl_{\textup{u}}$ over finite unary encoded strictly monotonic pure data words. 
        Like the $\ptime$-hardness proof for $\mtl$ in Theorem~\ref{thm-mtl-P}, the proof is by reduction from SAM2CVP. 
        Unlike the data word in the proof of Theorem~\ref{thm-mtl-P}, we encode 
        the wires between two levels of an SAM2-circuit by a \emph{strictly monotonic} data word. 
        This is possible due to the succinct usage of unions of intervals annotating the finally and globally modalities in the $\rmtl$ formula $\psi$. 
        The reader might look at Example~\ref{ex-circuit_succinct} below, where the construction is carried out for the circuit shown in Figure~\ref{fig-circuit}. 
        \begin{theorem} \label{theorem-P-lower-bound}
        	Path checking for $\pure\rmtl(\F,\X)_{\textup{u}}$ over finite unary encoded strictly monotonic pure data words is $\ptime$-hard.
        \end{theorem}
        \begin{proof}
        We reduce from SAM2CVP. Let $\alpha$ be an SAM2-circuit with $l \geq 2$ levels
        and $n$ gates in each level.
        The idea will be to encode the wires between two consecutive layers by a suitably shifted
        version of the  data
        word
        $$
        w_n = \prod_{i=1}^n i  \cdot \prod_{i=1}^n i (n+1) = (1,2,\ldots, n, \; 1 \cdot (n+1), 2 \cdot (n+1), \ldots, n \cdot (n+1)).
        $$
        Note that for all $i_1, i_2 \in \{1, \ldots, n\}$
        and $j_1, j_2 \in \{ 1 \cdot (n+1), 2 \cdot (n+1), \ldots, n \cdot (n+1)\}$ we  have the following:
        if $j_1 - i_1 = j_2 - i_2$, then $i_1 = i_2$ and $j_1 = j_2$. This is best seen by viewing numbers
        in their base $(n+1)$ expansion.
        Let us denote with $\Delta = n (n+1) - 1$ the maximal
        difference between a number from
        $\{1, \ldots, n\}$ and a number from $\{ 1 \cdot (n+1), 2 \cdot (n+1), \ldots, n \cdot (n+1)\}$.

       We define the pure and strictly monotonic data word $w_{n,l}$ as
       \begin{equation}
       	       \label{eqnarray_path_tptl1_dw}
        w_{n,l} = \prod_{j=0}^{l-2}  (w_n)_{+ j \cdot n(n+2)}
        \end{equation}
        The offset number $j \cdot n(n+2)$ is chosen such that the difference between a number from
        $\{1, \ldots, n\}$ and a number from $\{1 + j\cdot n(n+2), \ldots, n+ j\cdot n(n+2)\}$ is larger than
        $\Delta$ for every $j \geq 1$.

        Note that the unary encoding of the data word $w_{n,l}$ can be computed in logspace from the circuit.
        For each $j\in\{1,\dots, l-1\}$, define
        $$
        S_j=\{ i_2 (n+1) - i_1 \mid  \textrm{ the }i_{1}{\text{-th}} \textrm{ gate in level }
        j \textrm{ connects to}
         \textrm{ the } i_{2}{\text{-th}} \textrm{ gate in level }j+1\} .
        $$
        Suppose $o_k$, for some $k\in\{1,\dots, n\}$, is the designated output gate.  Let $I$ be the set of all $i \in [1,n]$ such that the $i$-th gate in layer $l$ is set to the
        Boolean value $1$.
        We construct the $\rmtl$ formula $\psi=\X^{k-1} \varphi_{1}$, where $\varphi_j$, for all $j\in\{1,\dots,l-1\}$, is defined inductively as follows:
        $$
        \varphi_{j}=\begin{cases}
        \F_{S_{j}} \X^n \varphi_{j+1} &\textrm{if } j< l-1 \textrm{ and level $j$ is a $\vee$-level,}\\
        \G_{S_{j}} \X^n \varphi_{j+1} &\textrm{if } j< l-1 \textrm{ and level $j$ is a $\wedge$-level,}\\
        \F_{S_{j}}(\bigvee_{i \in I} \X^{n-i} \neg \X \,\true) & \textrm{if } j= l-1 \textrm{ and level $j$ is a $\vee$-level,}\\
        \G_{S_{j}}(\bigvee_{i \in I} \X^{n-i} \neg \X \,\true) & \textrm{if } j= l-1 \textrm{ and level $j$ is a $\wedge$-level.}
        \end{cases}
        $$
        The purpose of the prefix $\X^n$ in front of $\varphi_{j+1}$ is to move from a certain position within
        the second half of the $j$-th copy of $w_n$ to the corresponding position within the first half
        of the $(j+1)$-th copy of $w_n$ in $w_{n,l}$.

        It is straightforward to check that $w_{n,l} \models \psi$ if, and only if,   the circuit $\alpha$ evaluates to~$1$.
              \end{proof}
              
              \begin{example} \label{ex-circuit_succinct}
              	      We illustrate the proof of Theorem~\ref{theorem-P-lower-bound} with the SAM2 from Figure~\ref{fig-circuit}. 
	The formula in equation \eqref{eqnarray_path_tptl1_dw} yields
	$$w_{5,3} =
        (1,2,3,4,5, \; 6,12,18,24,30) (36,37,38,39,40 \; 41,47,53,59,65). 
        $$
        We compute 
        \[S_1 = \{ 5,11,4,10,15, 21, 20,26, 13, 25\} \mbox{ and } S_2 = \{ 5,17,10,22,3,27,8,20,13,25\}.\] 
        Note that the values in $S_1$ correspond to all distances between gates in level 1 which are wired to gates in level 2. 
        Assuming that $a_3$ is the designated output gate and that $c_1,c_4,d_5$ ($c_2, c_3$, respectively) receive the input value $0$ ($1$, respectively), we have $I=\{2,3\}$ and we obtain the formula
       $$
       \psi = \X^2 \G_{S_1} \X^5\ \F_{S_2} (\X^3\neg\X\true \vee \X^2\neg\X\true).$$
                 \end{example}

        \noindent
        \emph{Proof of Theorem~\ref{corPhard}.} The theorem is a direct consequence of 
        	Theorem~\ref{theorem-P-lower-bound} and Proposition~\ref{fact2}.  \qed
        	
        	\medskip
        	\noindent
        	Similarly to the proof of Theorem~\ref{thm-mtl-P-infinite}, we can adapt the proof of Theorem~\ref{theorem-P-lower-bound}
        and extend the results to infinite data words.
        \begin{theorem}
        	Path checking for $\rmtl(\F,\X)_{\textup{u}}$ and $\tptl(\F,\X)^1_{\textup{u}}$ over infinite unary encoded data words are $\ptime$-hard.
        \end{theorem}	
\subsection{$\pspace$-Hardness Results}
Next we prove three $\pspace$ lower bounds, which complete the picture about the  complexity of the path-checking problem for $\mtl$ and $\tptl$.

\begin{theorem}\label{theorem-PSPACE-lower-bound}
	Path checking for $\pure\tptl(\F)_{\textup{u}}$ over finite unary encoded strictly monotonic pure data words is $\pspace$-hard.
    \end{theorem}
    \begin{proof}
    	    The proof is a reduction 
	    from the quantified Boolean formula problem: does a given closed formula $\Psi =Q_{1}x_{1}\dots Q_{n}x_{n}\phi$, where  $Q_{i}\in\{\forall,\exists\}$, and $\phi$ is a quantifier-free  propositional formula, evaluate to $\true$? 	    
	    This problem is $\pspace$-complete~\cite{GareyJohnson1979}.

	    Let $\Psi=Q_{1}x_{1}\dots Q_{n}x_{n}\phi$ be an instance of the quantified Boolean formula problem. 
        We construct the finite pure strictly monotonic data word
        \[
        w = 0,1,2, \ldots, 2n-1,2n,2n+1.
        \]
        For every $i\in\{1,\dots,n\}$, the subword $2i-1,2i$ is used to quantify over the Boolean variable $x_i$.

        For the formula, we use a register variable $x$ and register variables $x_i$ corresponding to the variables used in $\Psi$, for every $i\in\{1,\dots,n\}$. 
        Intuitively, if we assign to register variable $x_i$ the
        data value $2i-b$, then the corresponding Boolean variable $x_i$ is set to $b \in \{0,1\}$.

        We define the $\pure\tptl(\F)_{\textup{u}}$ formula $x.x_1.x_2.\dots x_n.\Psi'$,  where $\Psi'$ is defined inductively by the following rules.
        \begin{itemize}

          \item If $\Psi=\forall x_i \Phi$, then $\Psi'=\G ( (x_i = 2i-1 \vee x_i = 2i) \rightarrow x_i .\Phi' )$.
          \item If $\Psi=\exists x_i \Phi$, then $\Psi'=\F ( (x_i = 2i-1 \vee x_i = 2i) \wedge x_i .\Phi' )$.
          \item If $\Psi$ is a quantifier-free formula, then
         $$\Psi'=\F(x=2n+1 \wedge \Psi[x_{1}/x_{1}=2n,\dots,x_{i}/x_{i}=2(n-i)+2,\dots,x_{n}/x_{n}=2]).$$ Here, $\Psi[x_{1}/x_{1}=a_0,\dots,x_{n}/x_{n}=a_{n}]$ denotes the
         $\tptl$ formula obtained from $\Psi$ by replacing
         every occurrence of $x_i$ by $x_i = a_i$ for every $i\in\{1,\dots, n\}$.
           \end{itemize}
          Recall that the subformula $x_i = 2i-1 \vee x_i = 2i$  is \texttt{true} if, and only if, the difference
          between the current data value and the value to which $x_i$ is bound (and which, initially, is set to $0$) is $2i-1$ or $2i$.
          Hence, the subformula is only \true \ at the two positions where the data values are $2i-1$ and $2i$, respectively.
          Clearly, $\Psi$ is $\true$ if, and only if, $w \models_\tptl x.\Psi'$.
    \end{proof}
    
 The \emph{quantified subset sum problem}~\cite{Travers_2006_CMP_12387481238761} (QSS, for short) is to decide for a given sequence $a_{1},a_{2},\dots,a_{2n},b\in \N$ of binary encoded numbers, whether 
 $$\forall x_{1}\in\{0,a_1\}\exists x_{2}\in\{0,a_2\}\dots\forall x_{2n-1}\in\{0,a_{2n-1}\}\exists x_{2n}\in\{0,a_{2n}\} : \sum_{i=1}^{2n}x_{i} =b
 $$
 holds. This problem is $\pspace$-complete~\cite{Travers_2006_CMP_12387481238761}.

    	    We define a variant of QSS, called \emph{positive quantified subset sum problem} (PQSS, for short), 
    	    in which for a given sequence $a_{1},a_{2},\dots,a_{2n},b \in \mathbb{N} \setminus \{0\}$ of binary encoded numbers, we want to decide whether 
	    $$\forall x_{1}\in\{1,a_1\}\exists x_{2}\in\{1,a_2\}\dots\forall x_{2n-1}\in\{1,a_{2n-1}\}\exists x_{2n}\in\{1,a_{2n}\} : \sum_{i=1}^{2n} x_i =b .$$
    	    One can easily see that QSS and PQSS are polynomial time-interreducible, and thus PQSS is $\pspace$-complete.

    	    \begin{theorem}\label{theorem-PSPACE-lower-bound-2}
    	    	    Path checking for  $\pure\tptl^2(\F)_{\textup{b}}$ over the infinite strictly monotonic pure data word
    	    	    $w = 0 (1)^\omega_{+1} = 0,1,2,3,4,\ldots$   is $\pspace$-hard.
    \end{theorem}
    \begin{proof}
    	     The theorem is proved by a reduction from PQSS.
    	     Given an instance $a_{1},a_{2},\dots,a_{2n},b$ of PQSS, we construct the $\tptlpure^2(\F)_{\textup{b}}$ formula
        $x.\varphi_1$, where the formula $\varphi_i$, for every $i\in\{1,\dots,,2n+1\}$, is defined inductively by
        \[
        \varphi_{i}= \begin{cases}
        y.\G((y=1\vee y=a_i)\rightarrow \varphi_{i+1}) & \text{for $i <2n$ odd}, \\
        y.\F((y=1\vee y=a_i)\wedge \varphi_{i+1}) & \text{for $i \leq 2n$ even}, \\
        x=b & \text{for $i = 2n+1$}.
        \end{cases}
        \]

        The intuition is the following: note that in the data word $w$ the data value is increasing by one in each step.
        Assume we want to evaluate $y.\G((y=1\vee y=a_i)\rightarrow \varphi_{i+1})$ in a position where the data
        value is currently $d$. The initial freeze quantifier sets $y$ to $d$. Then, $\G((y=1\vee y=a_i)\rightarrow \varphi_{i+1})$
        means that in every future position, where the current data value is either $d+1$ (in such a position $y=1$ holds by the $\tptl$-semantics)
        or $d+a_i$ (in such a position $y=a_i$ holds), the formula $\varphi_{i+1}$ has to hold. In this way, we simulate
        the quantifier $\forall x_i \in \{1,a_i\}$. At the end, we have to check that the current data value is $b$, which can be done with the constraint $x=b$ (note that $x$ is initially set to $0$ and never reset).
        One
        can show that $(0)^\omega_{+1}\models_\tptl x.\varphi_{1}$ if, and only if, $(a_{1},a_{2},\dots,a_{2n},b)$ is a positive instance of PQSS.     	    
    \end{proof}

    \begin{theorem}\label{theorem-PSPACE-lower-bound-3}
         Path checking for $\fltl^2$ over infinite binary encoded pure data words is
         $\pspace$-hard.
 \end{theorem}
        \begin{proof}
        	The proof is by a reduction from PQSS.         	
        	We first prove the claim for infinite binary encoded data words. 
        	Then we show how one can change the proof to obtain the result for \emph{pure} data words.

        	Given an instance $a_{1},a_{2},\dots,a_{2n},b$ of PQSS, we define the infinite data word
        	\begin{align*}
        	w \defeq (r,b) \bigg( (q,0) \prod_{i=1}^{2n} (p, 1) (p, a_i) (r, 0)  \bigg)_{+1}^\omega
        	\end{align*}       
        For definining the $\fltl^2$ formula, we first define 
        for every $\fltl^2$ formula $\psi$ the auxiliary formulas
        \begin{align*}
        	\F_{p} \psi  \defeq p \U (p \wedge \psi) &  \text{ and }  \G_{p} \psi  \defeq   \neg \F_{p} \neg \psi .
      \end{align*}      
        The formula $\F_{p} \psi$ holds in position $i$ if, and only if, there exists a future position $j>i$ such that $\psi$ holds, and $p$ holds in  all positions $k\in\{i+1,\dots, j\}$. 
        Note that the formula $\G_p\psi$ is equivalent to $\neg p\R(p\to\psi)$; 
        it thus holds in a position $i$ if, and only if, for all future positions $j>i$, $\psi$ holds in position $j$ whenever $p$ holds in all positions $k\in\{i+1,\dots,j\}$.

        Define for every $i\in\{1,\dots,2n\}$ the $\fltl^2$ formula $\varphi_i$ as follows: 
        \begin{align*}
        	\varphi_i \defeq
        \begin{cases}
          \X^{ 3(i-1)} \G_{p} y. \F (q \wedge y=0 \wedge \varphi_{i+1})  & \text{for $i<2n$ odd}, \\
          \X^{3(i-1)}\F_{p} y. \F (q \wedge y=0 \wedge \varphi_{i+1})  & \text{for $i \leq 2n$ even}, \\
           x=0 & \text{for $i=2n+1$}.
        \end{cases}
        \end{align*}    
Finally set $\varphi=x.y.\X \varphi_1$.

Note that in $\varphi$, we require the auxiliary formulas $\F_p\psi$ and $\G_p\psi$ to hold in $w$ \emph{only} at positions in which also $q$ holds, \ie, at positions corresponding to the beginning of the periodic part of $w$. 
Plainly put, formula $\F_p\psi$ holds at some position in $w$ in which also $q$ holds, if, and only if,  $\psi$ holds at the next \emph{or} the next but one position; analogously, the formula $\G_p\psi$ holds, if, and only if, $\psi$ holds at the next \emph{and} the next but one position. 
        Note that we resign from using the next modality here to avoid an exponential blow-up of the formula.

        We explain the idea of the reduction. Assume $w\models_\tptl\varphi$.  
        Then, formula $\varphi_1$ holds at the beginning of the first iteration of the periodic part of $w$, \ie, at the position with letter $(q,0)$. 
        By formula $\G_py.(q\wedge y=0 \wedge \varphi_2)$, we know that 
        $y.(q\wedge y=0 \wedge \varphi_2)$ holds both at the position with letter $(p,1)$ \emph{and} at the position with letter $(p,a_1)$ (universal quantification by the $\G_p$-modality). 
        This is the case if, and only if, 
        $\varphi_2$ holds at the beginning of the $2$nd iteration of the periodic part of $w$ (\ie, the position with letter $(q,1)$) \emph{and}
        at the beginning of the $(a_1+1)$-th iteration of the periodic part of $w$ (the position with letter $(q,a_1)$). 
        In the former case, we can conclude that the formula $y.\F(q\wedge y=0\wedge \varphi_3)$ holds at the position with letter $(p,1+1)$ \emph{or} at the position with letter $(p,1+a_2)$ (existential quantification of the $\F_p$-modality); in the second case, the formula $y.\F(q\wedge y=0\wedge\varphi_3)$ holds at the position with letter $(p,a_1+1)$ \emph{or} at the position with letter $(p,a_1+a_2)$ (again, existential quantification of the $\F_p$-modality). Note how the intermediate sums computed so far, \ie, ($1+1$ or $1+a_2$, and, $a_1+1$ or $a_1+a_2$), are stored in the corresponding data values with proposition $p$.

        Note that the register variable $x$ is set to the first data value occurring in $w$, which is $b$. 
        Hence, in the formula $\varphi_{2n}$, the constraint $x=0$ expresses  that the current data
        value has to be $b$.

        Clearly, one can prove that $w \models_\tptl x.y.\X \varphi_1$ if, and only if, $a_{1},a_{2},\dots,a_{2n},b$ is a positive instance of PQSS.

        Next, we explain how we can encode the propositional variables occurring in $w$ and $\varphi$ by data values, to obtain the result for pure data words and the pure logics. 
        Note that $r$ does not occur in $\varphi$. 
        It thus suffices to encode $q$ and $p$, respectively, which we do by  $(0,1,1)$ and $(0,0,0)$, respectively. 
        We obtain the pure data word $w'$ from $w$ by replacing all occurrences of $p$ and $q$ by their corresponding data values, and by removing $r$, as follows, where we underline the data values that occurred in the $w$: 
        \begin{align*}
        	w' \defeq \underline{b}\, \bigg( 0,1,1,\underline{0}, \prod_{i=1}^{2n} (0,0,0, \underline{1}, 0,0,0, \underline{a_i}, \underline{0})  \bigg)_{+1}^\omega .
        	\end{align*}
    Define $\varphi_q = x.\X(\neg (x=0)\wedge x.\X(x=0))$ and $\varphi_p = x. \X(x=0 \wedge \X(x=0)) $. We replace the formula $\F_p \psi$ by
    \[
    \F_p' \psi = [\varphi_p \vee \X^3(\varphi_p\wedge \X \neg \varphi_p) \vee \X^2(\varphi_p\wedge \X \neg \varphi_p) \vee \X(\varphi_p\wedge \X \neg \varphi_p)]\U [\varphi_p \wedge \psi]
    \]
    and define $\G_p' \psi= \neg \F_p' \neg \psi$.
    Then we define:
        \[
        	\varphi_i' \defeq
        \begin{cases}
          \X^{9(i-1)}\, \G_{p}'\, \X^3 y. \F (\varphi_q  \wedge \X^4 \varphi_p \wedge \X^3 (y=0) \wedge \X^3\varphi_{i+1}')  & \text{for $i<2n$ odd}, \\
          \X^{9(i-1)}\, \F_{p}'\, \X^3 y. \F (\varphi_q  \wedge \X^4 \varphi_p \wedge \X^3 (y=0) \wedge \X^3\varphi_{i+1}')  & \text{for $i \leq 2n$ even}, \\
           x=0 & \text{for $i=2n+1$}.
        \end{cases}
        \]
        Analysing the formulas yields that $w' \models_\tptl x.y.\X^4\,\varphi_1'$ if, and only if,  
        $a_{1},a_{2},\dots,a_{2n},b$ is positive instance of PQSS.
        \end{proof}

\section{Model Checking for Deterministic One-Counter Machines}

\label{sec-one-counter}
A \emph{one-counter machine} (OCM, for short) is a nondeterministic finite-state machine extended with a single counter that takes values in the non-negative integers. 
Formally, a one-counter machine is a tuple $\A=(Q, q_0, E)$, where $Q$ is a finite set of control states, $q_0 \in Q$ is the initial control state, 
and $E\subseteq Q\times \Op \times Q$ is a finite set of labelled edges, where $\Op=\{\zero\}\cup\{\add(a)\mid a\in\Z\}$. 
We use the operation $\zero$ to test whether the current value of the counter is equal to zero, and we use $\add(a)$ for adding $a$ to the current value of the counter. 
A \emph{configuration} of the one-counter machine $\A$ is a pair $(q,c)$, where $q\in Q$ is a state and $c\in\N$ is the current value of the counter. 
We define a transition relation $\to_{\A}$ over the set of all configurations by $(q,c)\to_{\A}(q',c')$ if, and only if,  there is an edge $(q, \op, q')\in E$ 
and one of the following two cases holds: 
\begin{itemize}
\item $\op=\zero$ and  $c=c'=0$,
\item $\op=\add(a)$ and  $c'=c+a \geq 0$. 
\end{itemize}
If $\A$ is clear from the context, we write $\to$ for $\to_{\A}$.
A \emph{finite computation} of $\A$ is a finite sequence  $(q_0,c_0)(q_1,c_1)\dots(q_n,c_n)$ over $Q\times\N$ such that $c_0=0$ and $(q_i,c_i)\to(q_{i+1},c_{i+1})$ for all $i\in\{0,\dots,n-1\}$, and 
such that there does not exist  a configuration $(q,c)$ with $(q_n, c_n) \to (q, c)$.
We identify such a computation with the finite data word of the same form. 
An \emph{infinite computation} of $\A$ is an infinite sequence  $(q_0,c_0)(q_1,c_1)\dots$ over $Q\times\N$ such that, again, $c_0 = 0$, and $(q_i,c_i)\to(q_{i+1},c_{i+1})$ for all $i\ge 0$. 
We identify such a computation with the infinite data word of the same form. 
A {\em deterministic one-counter machine} $\A = (Q,q_0,E)$, briefly DOCM, is an OCM such that 
for every configuration $(q,c)$ with $(q_0,0) \to^* (q,c)$ 
there is \emph{at most one} configuration $(q',c')$ such that $(q,c)\to(q',c')$. 
This implies that $\mathcal{A}$ has a unique (finite or infinite) computation, which we denote by $\comp(\A)$, and which we view
as a data word as explained above. 
For complexity considerations, it makes a difference whether the numbers $a\in\Z$ in operations $\add(a)$ occurring at edges in $\A$ are encoded in unary or in binary. We will therefore speak of unary encoded (resp., binary encoded) OCMs in the following.
Let $\alpha(\A)$ be the largest number $a$ such that $\add(a)$ appears in an edge from $E$.
We use the following lemma from \cite{DBLP:journals/tcs/DemriLS10}

\begin{lemma}[\mbox{\cite[Lemma 9]{DBLP:journals/tcs/DemriLS10}}] \label{lemma-DLS}
Let $\mathcal{A}=(Q, q_0, E)$ be a DOCM. Then, the following holds:
\begin{itemize}
\item If $\comp(\A)$
	is infinite then $\comp(\A) = u_1 (u_2)^\omega_{+k}$ with
	$0 \leq k \leq |Q|$ and $|u_1 u_2| \leq \alpha(\A) \cdot  |Q|^3$.
\item  If  $\comp(\A)$
	is finite then  $|\comp(\A)| \leq \alpha(\A) \cdot |Q|^3$.
\end{itemize}
\end{lemma}

\begin{proof}
The first statement is shown in  \cite{DBLP:journals/tcs/DemriLS10} for a DOCM $\A$ with the operations
$\zero$, $\add(-1)$ and $\add(1)$ (and hence $\alpha(\A) = 1$). To get the full first statement of the lemma,
it suffices  to simulate every $\add(a)$ operation by at most $\alpha(\A)$ many operations $\add(-1)$ or $\add(1)$.
To the resulting DOCM one can then apply \cite[Lemma 9]{DBLP:journals/tcs/DemriLS10}.

The second statement
is implicitly shown in the proof of  \cite[Lemma 9]{DBLP:journals/tcs/DemriLS10}. Like for the proof of the first statement, it
suffices to consider a DOCM such that $a \in \{1,-1\}$ for all instructions $\add(a)$ and afterwards multiply the length of 
$\comp(\A)$  by $\alpha(\A)$. Assume that $\comp(\A)$ is finite and let $\comp(\A)=(q_0,c_0)(q_1,c_1)\dots(q_n,c_n)$.
Consider $i<j$ such that $c_i = c_j = 0$. By \cite[Lemma 10]{DBLP:journals/tcs/DemriLS10}, we must have $j-i \leq |Q|^2$. 
Moreover, there can be at most $|Q|$ many $i \geq 0$ such that $c_i = 0$. Let $k$ be maximal such that $c_k = 0$. 
From the previous discussion, we get $k \leq |Q|^2  (|Q|-1)$.

Assume that there exist $j > i > k$ with $q_i = q_j$. If $c_i \leq c_j$ then $\comp(\A)$ would be infinite,
and  if $c_i > c_j$ then the counter would hit zero again, which contradicts the choice of $k$.   It follows
that  $n \leq k + |Q|-1$  and hence $|\comp(\A)| = n+1 \leq k + |Q| \leq |Q|^2  (|Q|-1) + |Q| \leq |Q|^3$.
\end{proof}

For unary encoded DOCMs we make use of the following result:

\begin{lemma} \label{lemma-doca-to-data-word}
For a given unary encoded DOCM $\mathcal{A}=(Q, q_0, E)$ one can check in logspace, whether $\comp(\A)$ is finite or infinite.
Moreover, the following holds:
\begin{itemize}
\item If  $\comp(\A)$
	is finite, then the corresponding data word in unary encoding can be computed in logspace.
\item If $\comp(\A)$ is infinite, then one can compute in logspace two unary encoded data words
$u_1$ and $u_2$ and a unary encoded number $k$ such that $\comp(\A) = u_1 (u_2)^\omega_{+k}$.
\end{itemize}
\end{lemma}

\begin{proof}
In order to check whether $\comp(\A)$ is infinite, it suffices by Lemma~\ref{lemma-DLS}
to simulate $\A$ for at most $\alpha(\A) \cdot |Q|^3$ many steps. For this we store 
(i) the current configuration
$(q,c)$ with $c$ encoded in binary notation and (ii) a step counter $t$ in binary notation, which is initially zero and incremented
after each transition of $\A$. The algorithm stops if $(q,c)$ has no successor configuration or $t$ reaches the value $\alpha(\A) \cdot |Q|^3$. 
In  the latter case, $\comp(\A)$ is infinite.
Note that the counter value $c$ is bounded by 
$\alpha(\A)^2 \cdot |Q|^3$.
Therefore, logarithmic space suffices to store $(q,c)$ and $t$.
More precisely, $(q,c)$ can be stored with $4 \log |Q| + 2 \log \alpha(\A)$ bits (note that in the input representation, $\alpha(\A)$ is represented
with $\alpha(\A)$ bits) and $t$ can be stored with $3 \log |Q| + \log \alpha(\A)$ bits.
In a similar way, we can produce the data word $\comp(\A)$ itself in logarithmic space. We only have to print
out the current configuration in each step. Internally, our machine stores counter values in binary encoding. Since we want to output the 
data word in unary encoding, we transform the binary encoded counter values into unary encoding, which
can be done with a logspace machine. In case $\comp(\A)$ is finite, the machine outputs the unary encoding
of $\comp(\A)$ in this way. In case $\comp(\A)$ is infinite, we can output with a first logspace machine
the unary encoded data word consisting of the first $\alpha(\A) \cdot |Q|^3$ many configurations. From this data word,
a second logspace machine can easily compute two unary encoded data words
$u_1$ and $u_2$ and a unary encoded number $k$ such that $\comp(\A) = u_1 (u_2)^\omega_{+k}$.
We then use the fact that the composition of two logspace machines can be compiled into one logspace machine.
\end{proof}
In order to extend Lemma~\ref{lemma-doca-to-data-word} to binary encoded DOCMs, we need SLPs:

\begin{lemma} \label{lemma-doca-to-data-word-binary}
For a given binary encoded DOCM $\mathcal{A}$ one can check in polynomial time, whether $\comp(\A)$ is finite or infinite.
Moreover, the following holds:
\begin{itemize}
\item If  $\comp(\A)$
	is finite, then an SLP $\mathcal{G}$ with $\val(\mathcal{G}) = \comp(\A)$ can be computed
	in polynomial time.
\item If $\comp(\A)$ is infinite, then one can compute in  polynomial time  two 
	SLPs $\mathcal{G}_1$ and $\mathcal{G}_2$ and a binary encoded number $k$ such that 
  $\comp(\A) = \val(\mathcal{G}_1) (\val(\mathcal{G}_2))^\omega_{+k}$.
\end{itemize}
\end{lemma}
For the proof of Lemma~\ref{lemma-doca-to-data-word-binary} we need the following lemma:

\begin{lemma} \label{lemma-SLP-iteration}
	Let $u$ be a finite data word, and let $m \in \N$ and $k \in \Z$ be  binary encoded numbers such that 
$d + ik \geq 0$ for all data values $d$ occuring in $u$ and all $0 \leq i \leq m$. From $u$, $m$, and $k$
one can construct in polynomial time an SLP for the data word 
$\prod_{i=1}^m u_{+ik}$.
\end{lemma}

\begin{proof}
	We first consider the case $k > 0$. First, assume that $m = 2^n$ for some $n \geq 0$.
If $U_n = \prod_{i=1}^{2^n} u_{+ik}$ then we obtain the following recurrence: 
$$
U_0 = u_{+k} \text{ and } U_{n+1} = U_n (U_n)_{+2^n k} .
$$
This recurrence can be directly translated into an SLP.
Second, assume that $m$ is not necessarily a power of two and let
$m = 2^{n_1} + 2^{n_2} + \dots + 2^{n_l}$ be the binary expansion of $m$,
where $n_1 <  n_2 < \dots < n_l$.
Let $m_j = 2^{n_1} + 2^{n_2} + \dots + 2^{n_j}$ for $1 \leq j \leq l$.
If  $V_j = \prod_{i=1}^{m_j} u_{+ik}$ then we obtain the following recurrence: 
$$
V_1 = U_{n_1}  \text{ and } V_{j+1} = V_j  (U_{n_{j+1}})_{+ m_j k} . 
$$
Again, this recurrence can be directly translated into an SLP.

Let us finally show how to reduce the case $k < 0$ to the case $k > 0$ (the case $k=0$ is easier).
Assume that $k < 0$ and
let $v = (u_{+(m+1)k})^{\text{rev}}$, where $w^{\text{rev}}$ denotes the data word $w$ reversed. 
Then, for $l = -k>0$ we get the following identity:
$$
\prod_{i=1}^m u_{+ik} = (\prod_{i=1}^m v_{+il})^{\text{rev}} .
$$ 
From an SLP for $\prod_{i=1}^m v_{+il}$ it is easy to compute in polynomial time an 
SLP for $ (\prod_{i=1}^m v_{+il})^{\text{rev}}$: one just has to replace every right-hand side
of the form $BC$ by $CB$. This shows the lemma.
\end{proof}
\noindent
{\em Proof of Lemma~\ref{lemma-doca-to-data-word-binary}.}
Fix a DOCM $\A=(Q, q_0, E)$. 
We define below a procedure $\comp(q)$, where $q \in Q$. This procedure constructs an SLP
for the unique computation that starts in the configuration $(q,0)$. Hence the call $\comp(q_0)$
computes an SLP for $\comp(\A)$. The algorithm stores as an auxiliary data structure a graph $G$ with vertex set $Q \cup \{ f \}$ 
(where $f \not\in Q$) which initially is the empty graph and to
which edges labelled with data words are added. These data words will be represented by
SLPs.
The idea is that an edge $q \xrightarrow{u} p$ is added
if the data word $u$ represents a computation of $\A$ from configuration $(q,0)$ to configuration $(p,0)$.
The overall algorithm terminates as soon as (i) an edge to node $f$ is added or (ii) the graph $G$ contains a cycle.

The call $\comp(q)$ starts simulating $\A$ in the configuration $(q,0)$. Let $(p_i,c_i)$
be the configuration reached after $i \geq0$ steps (thus, $p_0 = q$). 
After at most $|Q|+1$ transitions, one of the following  situations has to occur:
\begin{itemize}
\item There exists $0 \leq i \leq |Q|$ such that that the computation terminates with $(p_i,c_i)$. We 
add the edge $q \xrightarrow{u} f$ to the graph $G$, where $u = (p_0,0) (p_1, c_1) \dots (p_i,c_i)$ and the call
$\comp(q)$ terminates.
\item There exists $1 \leq i \leq |Q|+1$ such that $c_i = 0$. We add the edge $q \xrightarrow{u} p_i$
to the graph $G$, where $u = (p_0,0)  (p_1, c_1) \dots (p_{i-1},c_{i-1})$ and call
$\comp(p_i)$.
\item $c_i > 0$ for all $1 \leq i \leq |Q|+1$ and there exists $1 \leq j < l \leq |Q|+1$ such that $p_j = p_l$ and
$c_j \leq c_l$. We 
add the edge $q \xrightarrow{u} f$ to the graph $G$, where $u = u_1 (u_2)^\omega_{+k}$,
$u_1 = (p_0,0)  (p_1, c_1) \dots (p_{j-1},c_{j-1})$, $u_2 = (p_j,c_j)  (p_{j+1}, c_{j+1}) \dots (p_{l-1},c_{l-1})$
and $k = c_l - c_j$, and  the call
$\comp(q)$ terminates.
\item $c_i > 0$ for all $1 \leq i \leq |Q|+1$ and there exists $1 \leq j < l \leq |Q|+1$ such that $p_j = p_l$ and
$c_j > c_l$. Let $k = c_l - c_j < 0$. We compute in polynomial time the binary encoding of the 
largest number $m \geq 0$ such that
$c_i + m k > 0$ for all $j \leq i \leq l-1$. That means that the computation of $\A$ starts from $(q,0)$ 
with the data word $u \prod_{i=0}^m v_{+ik}$, where
$$
u = (p_0,0) (p_1, c_1) \dots (p_{j-1}, c_{j-1}) \text{ and } v = (p_j,c_j)  (p_{j+1}, c_{j+1}) \dots (p_{l-1},c_{l-1}).
$$
Moreover, by the choice of $m$, 
there exists an $i \in [j, l-1]$ such that $c_i + (m+1) k \leq 0$. Let $i$ be minimal with this property and define
$$
w = (p_j,c_j+(m+1)k) \dots (p_{i-1},c_{i-1}+(m+1)k) .
$$
This implies that the computation of $\A$ starts from $(q,0)$ 
with the data word $u (\prod_{i=0}^m v_{+ik})  w$. Moreover, if $c_i + (m+1) k < 0$ the computation
terminates in the configuration $(p_{i-1},c_{i-1}+(m+1)k)$  and we add to $G$ an edge from $q$ to $f$
labelled with $u (\prod_{i=0}^m v_{+ik})  w$. On the other hand, if $c_i + (m+1) k = 0$ then we add to $G$
an edge from $q$ to $p_i$ labelled with $u (\prod_{i=0}^m v_{+ik}) w$ and call $\comp(p_i)$.
\end{itemize}
We did not make the effort to make the above four cases non-overlapping; ties are broken in an arbitrary way.
As mentioned above we start the overall algorithm with the call $\comp(q_0)$ and terminate as soon as 
(i) an edge to node $f$ is added or (ii) the graph $G$ contains a cycle. This ensures that there is a unique 
path in $G$ starting in $q_0$ that either ends in $f$ (in case (i)) or enters a cycle (in case (ii)).
In case (i) the data word $\comp(\A)$
is $u_1 u_2 \dots u_n$, where the data words $u_1, u_2, \ldots, u_n$ label the path from $q_0$ to $f$
(note that $u_n$ can be an infinite data word). 
In case (ii) the data word $\comp(\A)$ is $u_1 u_2 \dots u_n (u_{n+1} \dots u_m)^\omega_{+0}$, where 
the data words $u_1, u_2, \ldots, u_n$ label the path from $q_0$ to the first state  on the cycle,
and $u_{n+1}, \ldots, u_m$ label the cycle. Finally, note that all data words that appear as labels in $G$
can be represented by SLPs that can be computed in polynomial time. For the label
$u (\prod_{i=0}^m v_{+ik}) w$ this follows from Lemma~\ref{lemma-SLP-iteration}.
\qed

\medskip
\noindent
The algorithm from the above proof can be also used to check in polynomial time whether a given binary encoded OCM is
deterministic, which is not clear from our definition of DOCMs. 
We run the same algorithm as in the above proof but check in every step whether more than one successor configuration exists. 

Let $\logic$ be one of the logics considered in this paper.
The {\em model-checking problem for $\logic$ over unary (resp., binary) encoded DOCMs} 
asks for a given unary (resp., binary) encoded DOCM $\mathcal{A}$ 
and a formula $\varphi \in \logic$,
whether $\comp(\A) \models_{logic} \varphi$ holds.
Based on Lemma~\ref{lemma-doca-to-data-word} we can now easily show:

\begin{theorem} \label{lemma-reduction-infinite-path-DOCM}
	Let $\logic$ be one of the logics considered in this paper. The model-checking problem for $\logic$ over unary encoded DOCMs
	is equivalent with respect to logspace reductions to the path-checking problem for $\logic$ over infinite unary
encoded data words.
\end{theorem}

\begin{proof}
	The reduction from the model-checking problem for $\logic$ over DOCMs to the  path-checking problem for $\logic$ over infinite unary
encoded data words follows from Lemma~\ref{lemma-doca-to-data-word}. For the other direction take a unary encoded infinite data word
 $w = u_1 (u_2)^\omega_{+k}$ and a formula $\psi \in \logic$. 
 It is straightforward to construct (in logspace) from $u_1, u_2$ a unary encoded DOCM $\A$ such that
 the infinite sequence of counter values produced by $\A$ is equal to the sequence of data values in $w$ with an initial
 $0$ (that comes from the initial configuration $(q_0,0)$) added.
 Moreover, no state of $\mathcal{A}$ repeats among the first  $|u_1u_2|$ many positions in $\comp(\A)$.
 Hence, by replacing every proposition in $\psi$ by a suitable disjunction of states of $\A$  
 we easily obtain a formula 
 $\psi' \in \logic$ such that $w \models \psi$ if, and only if,  $\comp(\A) \models \psi'$.
 \end{proof}
By Theorem~\ref{lemma-reduction-infinite-path-DOCM}, the left diagram from Figure~\ref{fig-paths} also shows the complexity
results for $\tptl$-model checking over DOCMs.

Finally, for binary encoded DOCMs, Lemma~\ref{lemma-doca-to-data-word-binary} and 
Theorem~\ref{Path check TPTL-upper-PSPACE-SLP} directly imply the following result
($\pspace$-hardness follows by a reduction similar to the one from the proof of Theorem~\ref{lemma-reduction-infinite-path-DOCM}):

\begin{theorem} \label{lemma-reduction-infinite-path-DOCM-binary}
	The model-checking problem for $\tptl$ over binary encoded DOCMs is $\pspace$-complete.
\end{theorem}


\section{Summary and Open Problems}

Figure~\ref{fig-paths} collects our complexity results for path-checking problems. We use $\tptl^{\ge 2}$ to denote the fragments in which at least $2$ registers are used. 
We observe that in all cases whether data words are pure or not does not change the complexity. 
For finite data words, the complexity does not depend upon the encoding of data words (unary or binary), and for $\tptl$ and $\rmtl$, it does not depend on whether a data word is monotonic or not. 
In contrast to this, for infinite data words, these distinctions indeed influence the complexity: for binary encoded data words the complexity picture looks different from the picture for unary encoded or (quasi-)monotonic data words.

We leave open the precise complexity of the model-checking problem for $\mtl$ and $\tptl^1$ over DOCMs. Model checking $\mtl$ over data words that are represented by SLPs is $\pspace$-complete (and this even holds for $\ltl$, see~\cite{DBLP:conf/concur/MarkeyS03}), but the SLPs that result from DOCMs have a very simple form (see the proof of Lemma 6.2). This may be helpful in proving a polynomial time upper bound. 

Kuhtz proved in his thesis \cite{Kuh10} that the tree checking problem for $\textup{CTL}$ (i.e., the question whether
a given $\textup{CTL}$ formula holds in the root of a given tree)  can be solved in $\textup{AC}^2(\textup{LogDCFL})$. 
One might try to combine this result with the $\textup{AC}^1(\textup{LogDCFL})$-algorithm of 
Bundala and Ouaknine \cite{DBLP:conf/icalp/BundalaO14} for $\mtl$ path checking over monotonic data words. 
There is an obvious $\textup{CTL}$-variant of $\mtl$ that
might be called $\textup{MCTL}$, which to our knowledge has not been studied so far. 
Then, the question is whether the tree checking problem for $\textup{MCTL}$ is in $\textup{AC}^2(\textup{LogDCFL})$. Here
the tree nodes are labelled with data values and this labelling should be monotonic in the sense that if $v$ is a child
of $u$, then the data value of $v$ is at least as large as the data value of $u$. 

Similarly to $\mtl$, also $\tptl$ has a $\textup{CTL}$-variant. One might try
extend our polynomial time path checking algorithm for the one-variable $\tptl$-fragment to this $\textup{CTL}$-variant. 

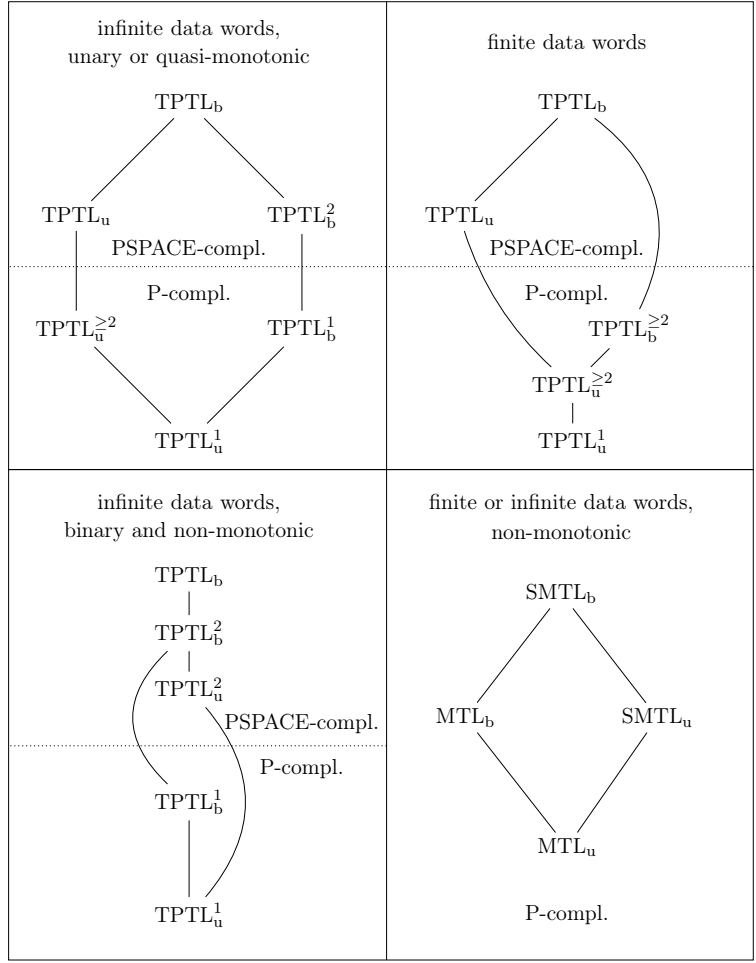
\begin{figure}[H]
\begin{center}
\setlength{\unitlength}{1mm}
 \scalebox{.75}{\begin{picture}(128,160)(0,-84)
    \gasset{Nframe=n,AHnb=0,Nadjust=wh,Nadjustdist=1}
    \node(1u-inf)(30,2){$\tptl^1_{\textup{u}}$}
    \node(ru-inf)(10,22){$\tptl^{\ge 2}_{\textup{u}}$}
    \node(1b-inf)(50,22){$\tptl^1_{\textup{b}}$}
    \node(u-inf)(10,42){$\tptl_{\textup{u}}$}
    \node(2b-inf)(50,42){$\tptl^2_{\textup{b}}$}
     \node(b-inf)(30,62){$\tptl_{\textup{b}}$}

     \node(1)(30,75){infinite data words,}
     \node(12)(30,70){unary or quasi-monotonic}
     \node(3)(30,-9){infinite data words,}
     \node(32)(30,-14){binary and non-monotonic}
     \node(2)(97,73){finite data words}
      \node(4)(30,28){$\ptime$-compl.}
      \node(5)(97,28){$\ptime$-compl.}
      \node(8)(50,-56){$\ptime$-compl.}
      \node(6)(30,36){$\pspace$-compl.}
      \node(7)(97,36){$\pspace$-compl.}
      \node(9)(50,-48){$\pspace$-compl.}


%

      \node(13)(97,-82){$\ptime$-compl.} 
      \node(14)(96,-9){finite or infinite data words,}
      \node(142)(96,-14){non-monotonic} 
      \node(sb)(96,-25){$\smtl_{\textup{b}}$} 
      \node(su)(113,-47){$\smtl_{\textup{u}}$} 
      \node(mb)(79,-47){$\mtl_{\textup{b}}$} 
      \node(mu)(97,-70){$\mtl_{\textup{u}}$} 
      \drawedge(sb,su){}
      \drawedge(sb,mb){}
      \drawedge(mb,mu){}
      \drawedge(su,mu){}

    \drawedge(1u-inf,ru-inf){}
    \drawedge(1u-inf,1b-inf){}
    \drawedge(ru-inf,u-inf){}
    \drawedge(u-inf,b-inf){}
    \drawedge(1b-inf,2b-inf){}
    \drawedge(2b-inf,b-inf){}

    \node(1u-f)(98,2){$\tptl^1_{\textup{u}}$}
    \node(ru-f)(98,12){$\tptl^{\ge 2}_{\textup{u}}$}
    \node(rb-f)(108,22){$\tptl^{\ge 2}_{\textup{b}}$}
    \node(u-f)(78,42){$\tptl_{\textup{u}}$}
    \node(b-f)(98,62){$\tptl_{\textup{b}}$}
    \drawedge(1u-f,ru-f){}
    \drawedge(ru-f,rb-f){}
    \drawedge[curvedepth=3](ru-f,u-f){}
    \drawedge(u-f,b-f){}
    \drawedge[curvedepth=-10](rb-f,b-f){}

    \drawrect[Nframe=y](-2,-3,130,80)
    \drawrect[Nframe=y](-2,-90,130,-3)
    \drawline(65,-90)(65,80)
    \drawline[dash={0.2 0.5}0](-2,33)(130,33)
    \drawline[dash={0.2 0.5}0](-2,-52)(65,-52)

    \node(1u-inf')(30,-82){$\tptl^1_{\textup{u}}$}
    \node(1b-inf')(30,-62){$\tptl^1_{\textup{b}}$}
    \node(2u-inf')(30,-42){$\tptl^2_{\textup{u}}$}
    \node(2b-inf')(30,-32){$\tptl^2_{\textup{b}}$}
     \node(b-inf')(30,-22){$\tptl_{\textup{b}}$}

    \drawedge(1u-inf',1b-inf'){}
    \drawedge[curvedepth=10](1b-inf',2b-inf'){}
    \drawedge[curvedepth=-10](1u-inf',2u-inf'){}
    \drawedge(2u-inf',2b-inf'){}
    \drawedge(2b-inf',b-inf'){}

 \end{picture}}
\end{center}
\caption{\label{fig-paths} Complexity results of path checking}
\end{figure}

%






\end{document}